\RequirePackage{fix-cm}
\documentclass{svjour3}                     
\smartqed  
\usepackage{graphicx}
\usepackage[cmex10]{amsmath}
\usepackage{amssymb}
\usepackage{amsmath}
\usepackage{mwe}
\usepackage{caption}
\usepackage{url}

\newcommand{\xbox}{\qed}

\newcommand{\after}{\text{-\ul{after}-}}

\newcommand{\fun}{\rightarrow}

\newcommand{\pass}{\ul{\text{\bf pass}}}

\newcommand{\ii}[1]{\ul{#1}}

\newcommand{\TS}{\text{TS}}

\newcommand{\ul}{\underline}
\newcommand{\ol}{\overline}

\begin{document}

\title{Complete Requirements-based Testing with Finite State Machines%
\thanks{Funded by the Deutsche Forschungsgemeinschaft (DFG, German Research Foundation) -- project number  407708394.}
}

\titlerunning{Complete Requirements-based Testing}        

\author{Wen-ling Huang         \and
        Jan Peleska
}
\offprints{Jan Peleska}


\institute{Wen-ling Huang \at
              University of Bremen \\
              Department of Mathematics and Computer Science\\
              \email{huang@uni-bremen.de}           
           \and
           Jan Peleska \at
              (Corresponding author) \\
              University of Bremen \\
              Department of Mathematics and Computer Science\\
              \email{peleska@uni-bremen.de}      \\
}


\maketitle

\begin{abstract}
In this paper, new contributions to requirements-based testing with deterministic finite state
machines are presented. Elementary requirements are specified as triples consisting of
a state in the reference model, an input, and  the expected reaction of the system under test   defined by a set of admissible outputs, allowing for different implementation variants. 
Composite requirements are specified as collections of elementary ones. 
Two   requirements-driven test generation strategies   
are
introduced, and their fault coverage guarantees are  proven. The first is exhaustive in the sense that it 
produces test suites guaranteeing requirements
satisfaction if the test suite is passed. If the test suite execution fails for a given implementation, however, 
this does not imply that the requirement has been violated. Instead, the failure may indicate an arbitrary violation
of I/O-equivalence, which could be unrelated to the requirement under test. The second strategy is complete in the sense that it produces test suites 
guaranteeing requirements satisfaction {\it if and only if} the suite is passed. Complexity considerations indicate that for practical application, 
the first strategy should be preferred to the second. 
Typical application scenarios for this approach are safety-critical systems, where safety requirements 
should be tested with maximal thoroughness, while user requirements might be checked with lesser effort, using conventional testing heuristics.
\end{abstract}

\section{Introduction}
\label{sec:intro}

\subsection{Background: Requirements-driven Model-based Testing}

In model-based black-box testing of embedded control systems, test suites can be generated from models with two alternative objectives in mind. 
(1) The test suite could aim at uncovering conformance violations; typical conformance relations are 
interface language equivalence or refinement. It is assumed that the model captures the complete expected behaviour
of the system under test (SUT).
(2) Alternatively, test suites can be constructed to uncover violations of specific requirements   in an implementation behaviour. To this end, the requirement has to be specified in addition to the model, using, for example, temporal logic or test scenario specifications. Another option is to use modelling languages like SysML~\cite{SysML17} allowing to relate requirements to behavioural or structural model elements, making it often unnecessary to add temporal logic or scenario specifications. Some requirements-driven test approaches advocate
test generation from temporal logic formulas alone~\cite{DBLP:conf/issta/WhalenRHM06}, so that a behavioural model becomes unnecessary. This, however, has the disadvantage that formulas can only refer to interface variables, making the formulas quite complex to specify. With a behavioural model at hand, formulas can also refer to internal state variables, facilitating the expression of requirements~\cite{DBLP:conf/isola/0001BH18}.

Objective (1)  typically applies to protocol testing or any other domain where the model is sufficiently small to induce conformance test suites of acceptable size~\cite{protocoltestsystems95}. Objective (2) applies to   domains where models are too large and too complex to perform conformance tests with acceptable size and sufficient test strength to uncover conformance violations. Moreover, development standards for safety-critical control systems typically require that testing should be requirements-driven, so that a conformance testing approach that does not relate test cases to requirements would be inadmissible~\cite{DO178C,CENELEC50128}. Finally, requirements-driven testing is the preferred approach to systematic regression testing: after changes to the implementation, the requirements affected by these changes are identified, and it is tested whether the modified SUT conforms to these requirements. This allows to avoid re-testing the whole test suite, which may be too time consuming, 
especially in HW/SW integration testing or system testing, where test executions need to be 
performed in physical time and cannot be sped up by using faster processors or performing them in parallel on many CPU cores, as is possible for software tests.

\subsection{Main Contributions}\label{sec:maincontrib}

In this article, two novel model-based   requirements-driven
test strategies are introduced, and their fault coverage guarantees are proven. 
Models are represented as deterministic finite state machines (DFSMs).
We are aware of the fact that FSMs are not well-suited for modelling control systems with complex and large data types for interfaces and internal state variables. 
We have shown, however,  that such systems may be abstracted to FSMs
after having calculated state and input equivalence classes on the complex model which is 
interpreted, for example, as a Kripke Structure. Test suites with guaranteed fault coverage calculated for the FSM abstraction give rise to    equivalence 
class tests with likewise guaranteed fault coverage for the complex system~\cite{peleska_sttt_2014,Huang2017}. These considerations motivate the study
of testing theories for finite state machines. We restrict ourselves to deterministic systems, since determinism is always 
required in the context of safety-critical control systems.

\emph{Elementary requirements} are represented as triples $R(q,x,Z)$, where $q$ is a state
in the reference model, $x$ is an input to the DFSM, and $Z$ is a subset of the machine's output alphabet, representing the admissible outputs that may be produced by the SUT after having processed
any input sequence leading the reference model into state $q$.  A \emph{composite requirement} is specified as a combination of elementary ones. It is easy to see that implementations whose true behaviour can be
represented by a DFSM which is    
language-equivalent to that of
the reference DFSM automatically fulfil all specified requirements.

The two strategies have the following characteristics.
\begin{itemize}
\item The first strategy is \emph{exhaustive} in the sense that it 
provides test suites that {\it imply} requirements satisfaction when passed. When failed, 
the SUT is guaranteed to violate language equivalence, but it does not necessarily violate 
the composite requirement that is being tested. This approach is called 
\emph{exhaustive testing of (composite) requirements}.

\item The second strategy is \emph{complete} in the sense that it
provides test suites that are passed by the SUT {\it if and only if} it conforms to the specified requirement. When a test suite generated according to this strategy  is failed, it is guaranteed that the SUT violates the   requirement. On the other hand, these test suites do not uncover any violations of language conformance, if these violations are unrelated to the requirement. We call this approach \emph{complete testing of (composite) requirements}.
\end{itemize}

Observe that the terms  `exhaustive' and `complete' have been adopted from conformance testing, where 
a test suite is called `sound' if it is passed by all conforming implementations, `exhaustive' if non-conforming implementations will always fail at least one test case, and  `complete' if the suite is 
sound and exhaustive~\cite{DBLP:journals/cn/Tretmans96}.

The test suite sizes depend on the requirement and on  the difference between the maximal number $m$ of states assumed for the SUT and the (known) number of $n$ of states in the minimised reference model.
For the first strategy, a detailed evaluation shows that   test effort reductions between 20\% and 60\% 
can be achieved   in comparison to 
test suites generated to prove language equivalence. We explain why 
the second strategy is of significant theoretical value,
but usually results in larger test suites and is of lesser practical value than the first.

The work presented here generalises previous 
publications~\cite{DBLP:conf/pts/Huang017,Huang2018}, where requirements referred to outputs of
different criticality only, but could not be linked to states and inputs of the reference model.

\subsection{FSM Library}

The test suites generated for the evaluation of the test strategies presented in this article  
have been   calculated using the 
\emph{fsmlib-cpp} library, an open source project programmed in C++. The library contains fundamental
algorithms for processing   Mealy Machine FSMs  
and a variety of model-based test generation algorithms. 
Download, contents, and installation of the library is explained in the lecture notes~\cite[Appendix B]{PeleskaHuangLectureNotesMBT} which are also publicly available.

\subsection{Overview}

In Section~\ref{sec:defs}, basic definitions about finite state machines that are
needed for the elaboration of results are presented.
The notion of elementary and composite requirements is introduced   
in Section~\ref{sec:req}. The soundness of the concept is justified by proving that language equivalence
can be alternatively expressed by composite requirements.
In Section~\ref{sec:dfsmabs}, we present two DFSM abstractions that are induced by our requirements notion
and needed for the construction of complete test suites.
In Section~\ref{sec:mainimplication}, a pass-relation for requirements-driven tests and 
the first main theorem about complete test suites implying requirements satisfaction are presented. 
The application of this main theorem to practical requirements-driven testing is illustrated 
and evaluated by means of several experiments in Section~\ref{sec:rwexample}.
In Section~\ref{sec:maintheorem}, the second main theorem yielding complete test suites to be passed by the SUT if and only if it fulfils the requirements is introduced and proven. 
The two test strategies are compared with respect to their complexity and to their practical value in Section~\ref{sec:complexity}. Section~\ref{sec:related}  discusses related work, and
Section~\ref{sec:conc} presents a conclusion. In Appendix~\ref{sec:fsmlib}, it is explained how the test suites described in examples throughout the article can be re-generated using the fsmlib-cpp.

\section{Basic Definitions}\label{sec:defs}

A \emph{finite state machine (FSM)} is a tuple $M = (Q,\ii q, \Sigma_I, \Sigma_O, h)$ with
finite state space $Q$, initial state $\ii q\in Q$, finite input and output alphabets
$\Sigma_I, \Sigma_O$, and transition relation 
$h\subseteq Q\times \Sigma_I \times \Sigma_O \times Q$. 
The \emph{specified inputs} of a state $q\in Q$ denote the set $in(q)$ of input elements
defined by 
$in(q) = \{x\in \Sigma_I~|~\exists y\in\Sigma_O,q'\in Q: (q,x,y,q')\in h \}$.

The \emph{language $L(q)$ of a state $q\in Q$} is the set of all   finite sequences
$\tau = (x_1,y_1)\dots (x_k,y_k)\in(\Sigma_I\times \Sigma_O)^*$ for which states  $q_1,\dots,q_k$   satisfying  
\begin{equation}\label{eq:tracecondition}
(q,x_1,y_1,q_1)\in h,\dots,(q_{k-1},x_k,y_k,q_k)\in h
\end{equation}
can be found.
The empty sequence $\varepsilon$ is also an element of $L(q)$.
 Sequence $\tau$ is called an \emph{I/O-trace} of $q$, and its 
length   is denoted by $|\tau| = k$. 
Sequence $\ol x = x_1\dots x_k$ is called
an \emph{input trace}, and $\ol y = y_1\dots y_k$ an \emph{output trace}.
The \emph{language $L(M)$ of the FSM $M$} is the language $L(\ii q)$ of its initial state.
We also use the alternative notations $\ol x / \ol y$ for $\tau$ and $x_i/y_i$ for
$(x_i,y_i)$ in $\tau$.  Trace segments of some trace $\tau$ from   $p^{th}$ element 
to   $q^{th}$ element
are denoted  by $\tau^{[p..q]}$. If $p < q$ as, for example, in $\tau^{[1..0]}$, this denotes 
the empty trace $\varepsilon$.

If $L(M) = L(M')$, the two machines are called \emph{I/O-equivalent} or \emph{language equivalent}; if $L(M')\subseteq L(M)$, the former is a \emph{reduction} of the latter.

An FSM is \emph{deterministic}, and the machine is called a DFSM, if and only if a pair $(q,x)\in Q\times \Sigma_I$ is 
associated with {\it at most} one transition element $(q,x,y,q')\in h$. The FSM is called \emph{completely specified} if and only if every pair  $(q,x)\in Q\times \Sigma_I$
is associated with {\it at least} one transition element $(q,x,y,q')\in h$. 

For completely
specified DFSMs, every pair  $(q,x)\in Q\times \Sigma_I$
is associated with {\it exactly} one transition element $(q,x,y,q')\in h$. In this case,
the transition relation $h$ can be represented by a \emph{transition function}
$\delta : Q\times \Sigma_I \fun Q$ and an \emph{output function} 
$\omega : Q\times \Sigma_I\fun \Sigma_O$, such that  $(q,x,y,q')\in h$ if and only if
$(q,x)\in Q\times \Sigma_I$ and
$(q,x,y,q') = (q,x,\omega(q,x),\delta(q,x))$. Transition functions and 
output functions can be extended to sequences
of inputs in a natural way by setting

\medskip
\begin{minipage}[c]{.49\textwidth}
\begin{eqnarray*}
\delta & : & Q\times\Sigma_I^*\fun Q
\\
\delta(q,\varepsilon) & = & q
\\
\delta(q,x.\ol x) & = & \delta(\delta(q,x),\ol x)
\end{eqnarray*}
\end{minipage}
\begin{minipage}[c]{.49\textwidth}
\begin{eqnarray*}
\omega & : & Q\times\Sigma_I^*\fun \Sigma_O^*
\\
\omega(q,\varepsilon) & = & \varepsilon
\\
\omega(q,x.\ol x) & = & \omega(q,x).\omega(\delta(q,x),\ol x)
\end{eqnarray*}
\end{minipage}
 
\bigskip
An FSM is \emph{observable} if and only if state, input and output uniquely determine the 
target state, that is, if the transition relation fulfils
$$
   \forall (q,x,y,q'), (q,x,y,q'') \in h: q' = q''.
$$
DFSMs are automatically observable, and each nondeterministic non-observable FSM can be transformed to a language-equivalent observable one~\cite{luo_test_1994}.

The \emph{prime machine} of a DFSM $M$ is a  DFSM $M'$ with minimal number of states, such 
that $M$ and $M'$ have the same language.  $M'$ is uniquely determined up to isomorphism. The same holds for observable  nondeterministic FSMs~\cite{Starke72}.

For a DFSM $M$ and a state $q\in Q$, an input trace $\ol x$
of an I/O-trace $\tau\in L(q), \tau=\ol x/\ol y$, leads to a uniquely 
determined target state $q_k = \delta(q,\ol x)$, 
since there is exactly one state sequence $q_1,\dots,q_k$
and output trace $\ol y$
satisfying \eqref{eq:tracecondition}. We use notation $q\after\ol x = q_k$. Note that
this is a partially defined operator, if the DFSM is not completely specified. We extend its
domain with the empty input trace $\varepsilon$ by setting $q\after\varepsilon = q$.
We extend the $\after$ operator to set-valued right-hand side operands by setting 
$q\after X = \{ q\after\tau~|~\tau\in X\}$ for $X\subseteq \Sigma^*$.

A \emph{state cover} of a DFSM $M$ is a set $V \subseteq\Sigma_I^*$ of input traces 
such that for each
$q\in Q$, there exists a $v\in V$ satisfying $\ii q\after v = q$. In this article, it is
always required that $\varepsilon\in V$, because the empty trace reaches the initial state, 
i.e.~$\ii q\after\varepsilon = \ii q$. The following well-known 
lemma shows how a state cover can be
constructed in a very basic (though not optimal) way. It will be used in the proof of 
the main theorem.

\begin{lemma}
\label{lemma:statecov}
Let $M=(Q,\ii q,\Sigma_I,\Sigma_O,h)$ be an   
FSM over input alphabet $\Sigma_I$ and output alphabet $\Sigma_O$. Let $V\subseteq \Sigma_I^*$ 
be a finite  set of input traces containing the empty trace $\varepsilon$.  
Then either
\begin{enumerate}
\item $\ii q\after V$ contains all reachable states, i.e.,  
$\ii q\after V = \ii q\after\Sigma_I^*$, or 

\item 
$\ii q\after (V\cup V.\Sigma_I)$ contains at least one additional state which
is not contained in $\ii q\after V$, that is, 
$\ii q\after V\subsetneq\ii q\after (V\cup V.\Sigma_I)$.
\end{enumerate}
\end{lemma}
\begin{proof}
See, for example,  \cite[Lemma~4.2]{PeleskaHuangLectureNotesMBT}.
\xbox
\end{proof}

\section{Requirements}\label{sec:req}

\subsection{Requirements and Acceptable Deviations -- Motivation}

In this section, elementary requirements are introduced as assertions stating that specific
transitions of the reference model need to be performed correctly, whenever the SUT resides 
in a state corresponding to the transition's source state in the reference model. 
While, just as in the case of conformance testing, the correct behaviour is specified by the reference model, an \emph{acceptable deviation} is also specified for each requirement. Such
a deviation consists of a set of outputs that are considered as ``harmless'' 
for the specified state and input pair, though the correct output specified in the reference model is still the preferred result. 

\begin{example}
\label{ex:monitor}
Consider a speed monitor in a train that automatically triggers the service brakes when the train exceeds the admissible velocity. If the emergency brakes are triggered instead of the service brakes, this would be considered as an acceptable deviation: the safety goal to slow down the train will be reached even better than with the service brakes. Only the passenger comfort is reduced, since the emergency brakes are much harder than the normal ones. On the other hand, if the brakes are released instead of triggering them, this would lead to a safety violation and therefore be an unacceptable deviation.\xbox
\end{example}

The reason to introduce acceptable deviations is that the test effort can be reduced in a considerable way if (a) only critical requirements are tested with complete strategies, instead of using a complete conformance suite, and (b) we abstain from insisting on the {\it exact} behaviour specified in the reference model, as long as the resulting deviations are still acceptable.

\subsection{Elementary Requirements, Acceptable Deviations, and Satisfaction}

Formalising this intuitive concept, let $M=(Q, \ii q, \Sigma_I, \Sigma_O, h)$ be a reference model, represented by a deterministic, completely specified prime FSM
with $|Q|=n\ge 2$. Let $\omega_M:Q\times \Sigma_I\rightarrow \Sigma_O$ denote the output function of $M$. 
An \emph{elementary requirement} is denoted by $R(q,x,Z)$, where 
\begin{enumerate}
\item $q\in Q$,
\item $x\in \Sigma_I$, and
\item $Z\subset \Sigma_O$, such that $Z \neq \Sigma_O$ and $\omega_M(q,x)\in Z$. 
\end{enumerate}
Output $\omega_M(q,x)$ is called the \emph{expected output}, and the outputs $y \in Z-\{ \omega_M(q,x)\}$ are called \emph{acceptable deviations}. The set $Z$ is required to be
a true subset of the output alphabet, because it would not make sense to specify a requirement
where {\it any} output $y\in\Sigma_O$ is considered as acceptable.

An elementary requirement $R(q,x,Z)$ is \emph{fulfilled with acceptable deviations} 
(written $S\models R(q,x,Z)$)
by some SUT $S = (S,\ii s,\Sigma_I, \Sigma_O, h_s)$, if, after any input trace $\ol x\in \Sigma_I^*$ leading to $q$ in the reference model, the input
$x$ given to $S$ results in an output $z\in Z$. To formalise this notion, define
\[
\Pi(q) = \{ \ol x\in\Sigma_I^*~|~\ii q\after\ol x = q \};
\]
this is the (possibly infinite) set of all input traces reaching $q$ in 
reference model $M$, when starting from the initial state 
$\ii q$. Then specify
\begin{equation}
S \models R(q,x,Z)\ \text{if and only if}\ 
\forall \pi\in\Pi(q): \omega_S(\ii s\after\pi,x) \in Z.
\end{equation}

\subsection{Composite Requirements for DFSMs}

A \emph{composite requirement $R$} is written as a conjunction
\[
R\equiv R(q_1,x_1,Z_1)\wedge \dots\wedge R(q_k,x_k,Z_k)
\]
of elementary requirements and interpreted in the natural way as
\[
S\models R \ \text{if and only if}\ (S\models R(q_1,x_1,Z_1)) \wedge\dots\wedge
 (S\models R(q_k,x_k,Z_k)).
\]
This means that the SUT $S$ fulfils each of the elementary requirements involved.

\begin{example}\label{ex:runninga}
As a running example, we will consider a reference model represented by
 the state machine $M$ shown in Fig.~\ref{fig:refmodelM}. The initial state $q_0$ 
 is marked by a double circle. On $M$, we specify the composite requirement
 \[ 
 R \equiv R(q_0,a,\{0,1\}) \wedge R(q_1,b,\{ 0,2\}) \wedge R(q_2,a,\{0,1\}).
 \]

 Consider an implementation with behaviour as modelled by the DFSM $S$ in Fig.~\ref{fig:sutmodelS}.
 The implementation is not language equivalent to $M$: for example, $\omega_M(q_0,a.a.a) = 1.0.0$, but
 $\omega_S(s_0,a.a.a) = 1.0.1$. However, the following example traces show that $S$ might still fulfil $R$ with acceptable deviation\footnote{This will be established more formally in the consecutive sections, where a complete test strategy for checking 
 the implementation's compliance to composite requirements is elaborated.}: (a) For sample input trace $b.b\in\Pi(q_0)$,
 we observe that
 $\omega_S(s_0\after b.b,a) = \omega_S(s_0,a) = 1 = \omega_M(q_0,a) = \omega_M(q_0\after b.b,a)$, 
 so  this result conforms to 
 $R(q_0,a,\{0,1\})$ without any deviation.
 (b) For $a.b\in\Pi(q_1)$, we get
 $\omega_S(s_0\after a.b,b) = \omega_S(s_1,b) = 0 = \omega_M(q_1,b) = \omega_M(q_0\after a.b,b)$, so  this result also 
 conforms to the applicable elementary requirement
 $R(q_1,b,\{0,2\})$ without any deviation.
 (c) For $a.a\in\Pi(q_2)$,    $\omega_S(s_0\after a.a,a) = \omega_S(s_0,a) = 1$, whereas
 $ \omega_M(q_0\after a.a,a) = \omega_M(q_2,a) = 0$. But this deviation is still acceptable, since the applicable elementary requirement   $R(q_2,a,\{0,1\})$ accepts both outputs $0$ and $1$.
 \qed
\end{example}

\begin{minipage}[t]{.48\textwidth}
\centering
 \includegraphics[width=0.7\textwidth]{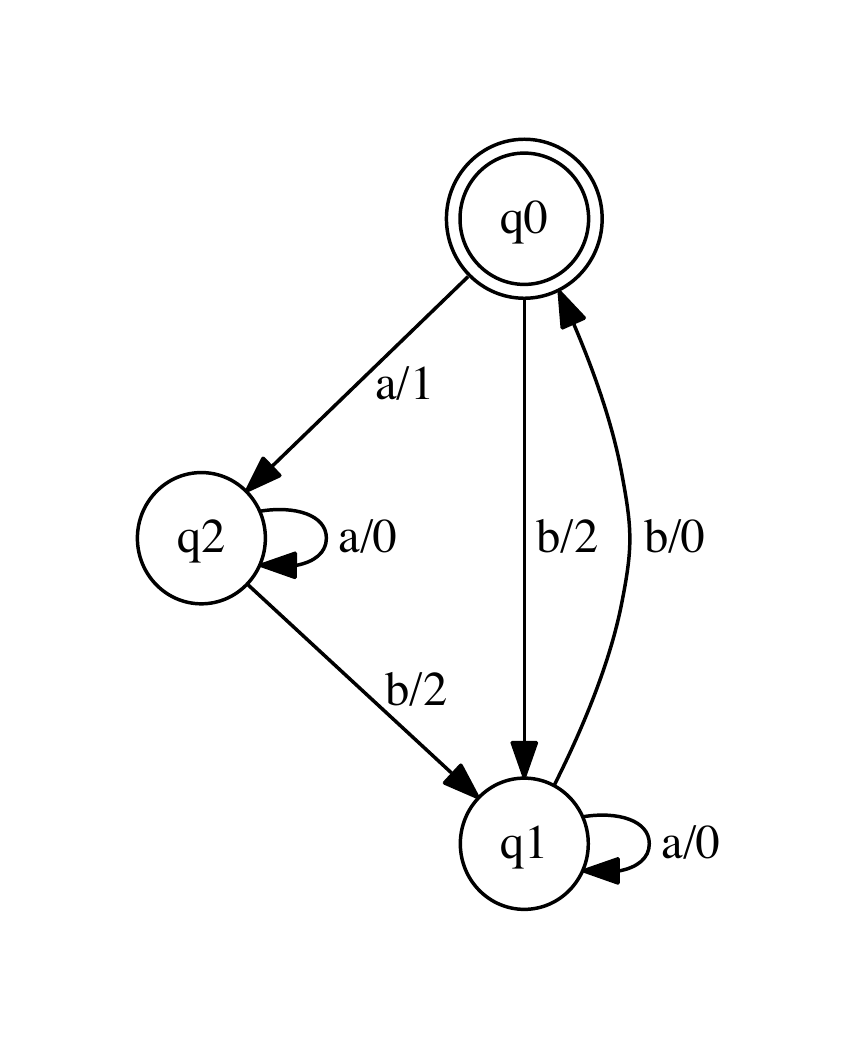}
 \captionof{figure}{Reference model $M$.}
\label{fig:refmodelM}
\end{minipage}
{~}
\begin{minipage}[t]{.48\textwidth}
\centering
 \includegraphics[width=0.55\textwidth]{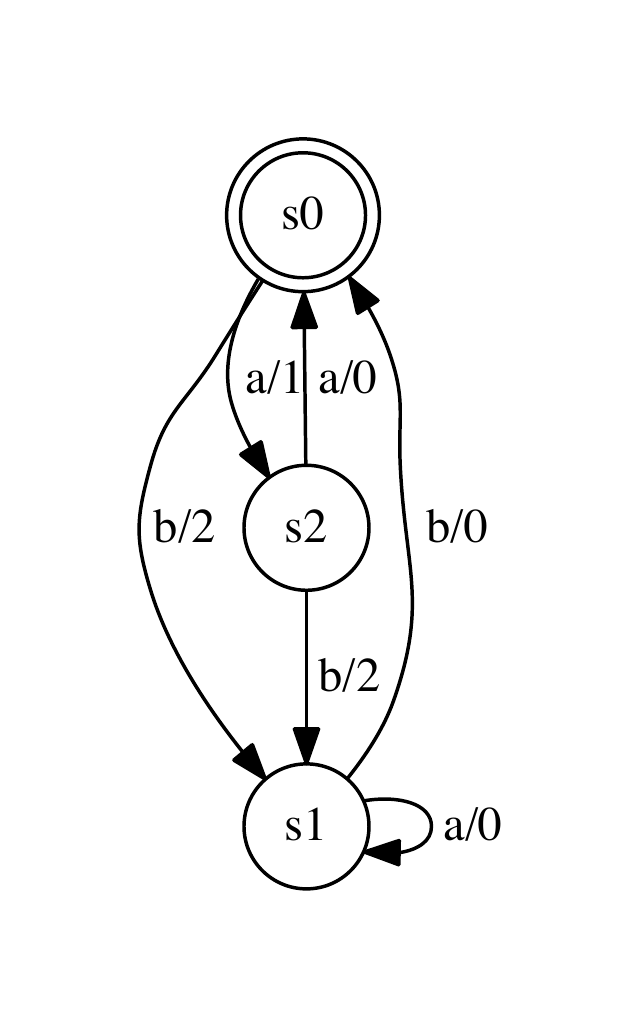}
 \captionof{figure}{Implementation model $S$.}
\label{fig:sutmodelS}
\end{minipage}

\subsection{Characterisation of Language Equivalence  by Composite Requirements.}

A ``very large'' composite requirement can be defined by stating that the SUT should conform to 
{\it every} transition of the reference model, without leaving any alternatives for the expected output.
The following theorem states that this requirement exactly captures language equivalence.

\begin{theorem}\label{th:languageeq}
Let $M=(Q,\ii q, \Sigma_I, \Sigma_O,h_M)$ and  $S=(S,\ii s,\Sigma_I,\Sigma_O,h_S)$
be deterministic, completely specified 
prime machines over the same alphabet with output functions
$\omega_M:Q\times \Sigma_I\rightarrow {\Sigma_O}$  and 
$\omega_S:Q\times \Sigma_I\rightarrow {\Sigma_O}$, respectively.
Define composite requirement 
$$
R_{eq}=\bigwedge_{(q,x)\in Q\times \Sigma_I}R(q,x, \{\omega_M(q,x)\}).
$$ 
Then 
$$
L(S)=L(M)\quad\text{if and only if}\quad S\models R_{eq}.
$$
\end{theorem}
\begin{proof}
It is obvious that $L(S)=L(M)$ implies $S\models R_{eq}$, 
since $M\models R_{eq}$ by construction of $R_{eq}$.
Therefore, we only need   to show that $S\models R_{eq}$ implies $L(S)=L(M)$.
Suppose $L(S)\neq L(M)$. Then there exists  a shortest input sequence $\overline x\in \Sigma_I^*$ such that $\omega_S(\ii s,\overline x)\neq\omega_M(\ii q,\overline x)$, where $\omega_S, \omega_M$ are the natural extension of the output functions to input sequences 
introduced in Section~\ref{sec:defs}.  Since $\varepsilon\in L(M)\cap L(S)$, $|\overline x|\ge 1$ and $\overline x$ can be written as $\overline x=\pi.x$, for some $\pi\in \Sigma_I^*$ and $x\in \Sigma_I$. Then we have $\omega_S(\ii s,\pi)=\omega_M(\ii q,\pi)\wedge \omega_S(\ii s\after\pi,x)\neq \omega_M(\ii q\after \pi, x)$. 
Hence $S\not\models R(\ii q\after \pi, x, \{\omega_M(\ii q\after \pi, x)\})$, and this elementary requirement
is a conjunct of the composite requirement $R_{eq}$.
Therefore, 
$S\not\models R_{eq}$ follows.
\xbox
\end{proof}

\section{Requirements-driven DFSM Abstraction}
\label{sec:dfsmabs}

Given a reference model $M$ and associated elementary requirements 
$R(q_i,x_i,Z_i)$ with $i=1,\dots,k$
 as introduced above, we will now introduce two abstractions 
$M \rightarrow M_1 \rightarrow M_2$ that are needed to create test suites
allowing to verify that an SUT $S$ fulfils all $R(q_i,x_i,Z_i)$ without having to test for the 
stronger property language equivalence.
\begin{enumerate}
\item $M_1$ abstracts from concrete $M$-outputs by using sets of output events
$Z_i\subset \Sigma_O$ where requirement $R(q_i,x_i,Z_i)$ is involved, and using the
whole output alphabet $\Sigma_O$ as {\it don't care symbol} to specify outputs unrelated to
any requirement.

\item $M_2$ is the prime machine of $M_1$. Note that $M_1$ may be no longer  minimal, 
since some
$M$-states may not be distinguishable anymore due to output abstraction.

\end{enumerate}

\subsection{Construction of $M_1$}
Let  $*=\Sigma_O$ denote the \emph{don't care symbol} specifying that ``any output is allowed''  
in certain situations. Define $\Sigma_O'=\{*, Z_1, \dots, Z_k\}$ as the new output alphabet of a
completely specified, deterministic abstraction $M_1$ with state space $Q$, 
initial state $\ii q$, 
input
alphabet $\Sigma_I$, and output and transition functions
specified as follows.
\begin{enumerate}
\item The transition function $\delta_{M_1}$ of $M_1$ coincides with  that of $M$, 
that is, $\delta_{M_1} = \delta_M$. 
\item The output function $\omega_{M_1}$ of $M_1$ is defined by
\begin{align}
\omega_{M_1}(q_i,x_i) & =  Z_i &\ \text{for}\ i\in \{1,\dots,k\}
\\
\omega_{M_1}(q,x) & =  * &\ \text{for}\  (q,x)\in Q\times \Sigma_I
                         \setminus \{ (q_1,x_1),\dots,(q_k,x_k) \}
\end{align}
\end{enumerate}
By construction, $M_1$ abstracts all outputs related to elementary 
requirements $R(q_i,x_i,Z_i)$ to
$Z_i$ and all outputs that are unrelated to any requirement to $*$. Being unrelated to any requirement means that the (state,input)-pair $(q,x)$ occurring in the 
original transition $(q, x, y, q')\in h$
differs from all $(q_i,x_i)$ used in 
the specification of some requirement $R(q_i,x_i,Z_i)$.

\subsection{Construction of $M_2$} \label{sec:m2}
Let $M_2$ be the prime machine of $M_1$. If $\{q_1,\dots,q_k\}=Q$ and $(x_i,Z_i)=(x_j,Z_j)$,
for all $i,j=1,\dots, k$, then $M_2$ contains only one state, otherwise, $M_2$ contains at least two states.
Denote the states in $M_2$ by $[q]$, which is the equivalence class of state $q$ in $M_1$ : 
$$
[q]=\{q'\in Q~|~L_{M_1}(q')= L_{M_1}(q)\}
$$
The transition relation of $M_2$ is denoted by $h_2=\{([q],x,y,[q'])~|~(q, x, y, q')\in h_1\}$. We have $[q\after \ol x]=[q]\after \ol x$, for any $q\in Q$ and $\ol x\in \Sigma_I^*$, where the transition "$\after$" from the left side is according to $h$ and $h_1$, since $M_1$ contains the same transitions of $M$ up to the outputs. The "$\after$" from the right side is according to $h_2$. 

Let $n'=|Q'|\ge 1$.   
For any $q_i,q_j\in \{q_1,\dots,q_k\}$, $[q_i]=[q_j]$ implies $x_i/Z_i=x_j/Z_j$. 
Let $q\in Q$ and suppose $[q]=[q_i]$ for some $i=1,\dots, k$. Then $(q,x_i,Z_i)\in \{(q_1,x_1,Z_1),\dots, (q_k,x_k,Z_k)\}$, because otherwise the output of 
$(q, x_i)$ in $h_1$ is $*$ and $*\neq Z_i$, contradicting the fact that 
$q$ and $q_i$ are equivalent. Hence, from $[q]=[q_i]$ follows $q=q_j$ for some $j=1,\dots, k$. 
Define 
$$
[\Pi(q_i)]=\bigcup_{j\in\{\ell\in\{1,\dots,k\}~|~[q_\ell]=[q_i]\}}\Pi(q_j).
$$ 
This is the set of all input traces leading to states $q_j$ in $M$ that are equivalent 
to $q_i$ in $M_1$.

\begin{example}\label{ex:m1construction}
For the reference model $M$ introduced in Example~\ref{ex:runninga} and the composite requirement $R$ specified there,
the abstracted machine $M_1$ is depicted in Fig.~\ref{fig:abstractedM1}. Its minimised machine $M_2$ is shown in
Fig.~\ref{fig:minimisedM1}.
\qed
\end{example}

\begin{minipage}[t]{.48\textwidth}
\centering
 \includegraphics[width=0.7\textwidth]{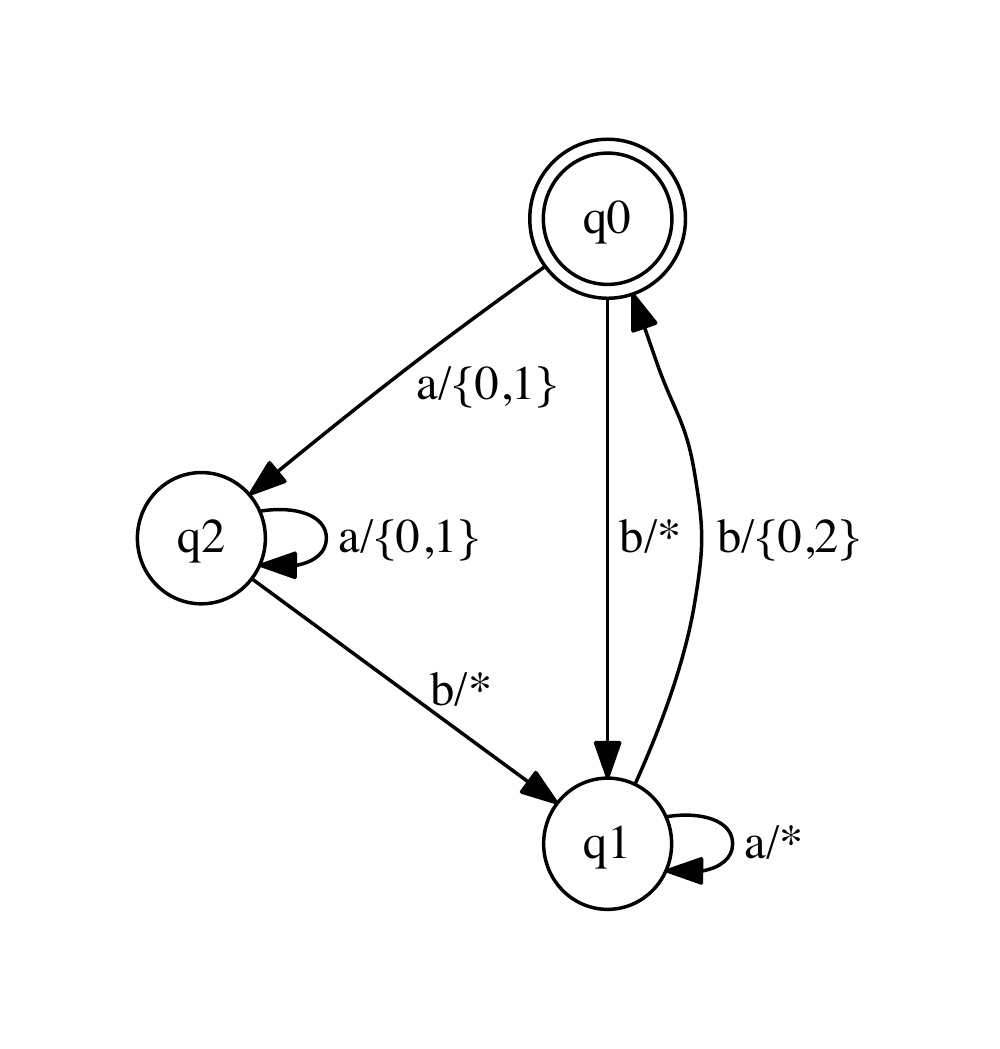}
 \captionof{figure}{Abstraction $M_1$ of reference model $M$ from Example~\ref{ex:runninga}.}
\label{fig:abstractedM1}
\end{minipage}
{~}
\begin{minipage}[t]{.48\textwidth}
\centering
 \includegraphics[width=0.55\textwidth]{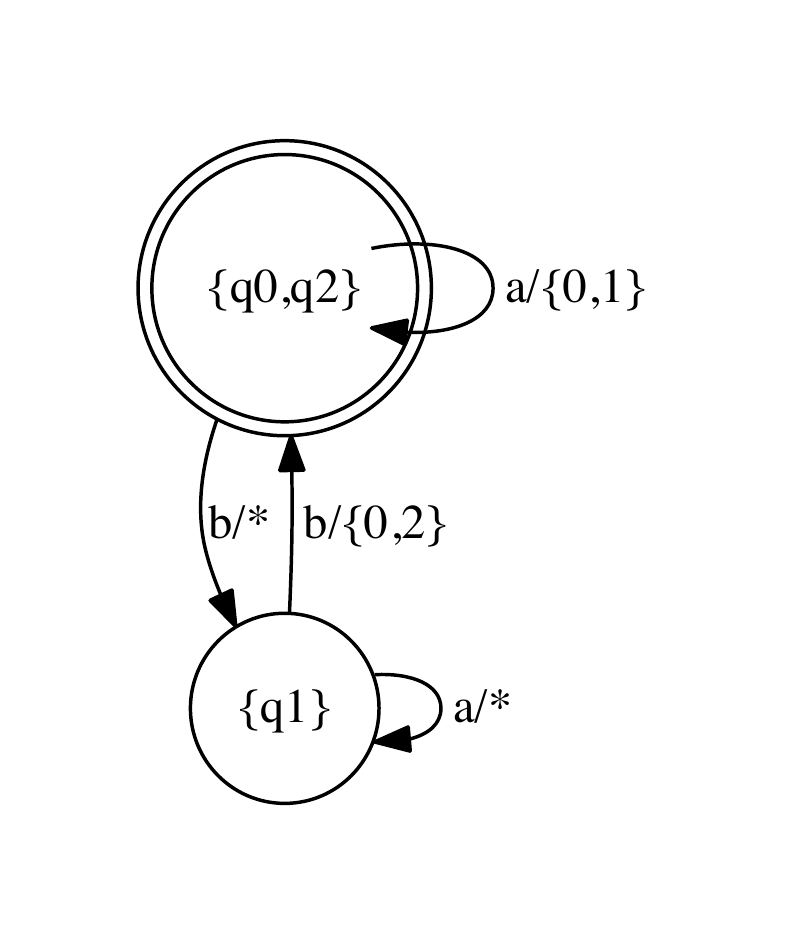}
 \captionof{figure}{Minimised version of $M_1$.}
\label{fig:minimisedM1}
\end{minipage}

\section{Exhaustive Testing of Composite Requirements}
\label{sec:mainimplication}

\subsection{Test Cases, Test Suites, and Pass Criteria}
A \emph{test suite} $\TS$ is a set of input traces $\ol x\in\Sigma_I$; the latter are called \emph{test cases}. 
The \emph{expected result} associated with test case $\ol x$ is the output trace
$\omega_M(\ii q,\ol x)$ calculated from the reference model.
An implementation 
$S=(S,\ii s,\Sigma_I,\Sigma_O,h_S)$ \emph{passes} test case $\ol x$, if and only if 
the associated outputs observed when running $\ol x$ against $S$ conform to the expected results. More formally,
\[
S\ \pass_\Rightarrow\ \ol x \ \ \text{if and only if}\ \ \omega_S(\ii s,\ol x) = \omega_M(\ii q,\ol x).
\]
Implementation $S$ \emph{passes the test suite $\TS$} if and only if $S$ passes all test cases, that is,
\[
S\ \pass_\Rightarrow\ \TS \  \ \text{if and only if}\ \  \forall \ol x:\TS: S\ \pass_\Rightarrow \ol x.
\] 
We say that a test suite $\TS$ is \emph{exhaustive with respect to composite requirement $R = R(q_1,x_1,Z_1)\wedge \dots \wedge R(q_k,x_k,Z_k)$} if and only if
\begin{equation}\label{eq:suffcomplete}
S\ \pass_\Rightarrow\ \TS \Rightarrow S\models R.  
\end{equation}

In the remainder of this section, we will construct a finite test suite $\TS_\Rightarrow$, whose size depends on the reference model $M$, the (composite) requirement,   and the assumed maximal number $m$ of states in the
implementation DFSM $S$. We will show that $\TS_\Rightarrow$ is exhaustive 
with respect to composite requirement $R = R(q_1,x_1,Z_1)\wedge \dots \wedge R(q_k,x_k,Z_k)$ 
for the fault domain of completely specified, deterministic FSMs with at most $m$ states.

Obviously, if $S$ is I/O-equivalent to $M$, then $S$ satisfies the composite requirement $R$. Hence, any SUT $S$, with $S\not\models R$  will fail 
any complete test suite for fault model $F=(M,\sim, D)$, where $S\in D$ and $\sim$ is the language equivalence conformance relation. Therefore, any complete test suite for fault model $F=(M,\sim, D)$ with $S\in D$ satisfies the above two conditions.
Our objective is to introduce a test strategy leading to {\it fewer} test cases than the 
well-known strategies for generating test suites checking language equivalence.

It should be emphasised that the expected results associated with each test case reject any deviation of the SUT from the outputs expected according to the reference model. Therefore, if $S$ is in a state corresponding to some state $q_i$ in $M$ and reacts to input $x_i$ with some
output $y$ which is in $Z_i$ but which differs from the expected value $\omega_M(q,x_i)$, the
test execution will fail. It will be illustrated below, however, that the introduction of
admissible deviation $y\in Z_i-\{ \omega_M(q,x_i) \}$ may lead to fewer test cases. This 
reduction comes at the cost that some violations of expected results may be overlooked; but
it can still be guaranteed that these violations are always admissible deviations. For transitions that are not linked to any elementary requirement $R(q_i,x_i,Z_i)$, no guarantees are made whatsoever. 

Moreover, note that according to condition \eqref{eq:suffcomplete}, failing the test suite does not necessarily imply that $R$ is violated: the test suite may also fail because 
an erroneous output is uncovered which is unrelated to any elementary requirement in $R$. This conforms to our understanding of a ``reasonably designed'' test oracle. Any detected deviation in comparison to the reference model will lead to the test case to fail; it is just not guaranteed that {\it all} errors that are unrelated to $R$ will be uncovered.

\subsection{Test Suite Construction}

Let $V$, $\varepsilon\in V$, be a state cover of $M$. 
We define three auxiliary sets $A, B, C$ containing pairs of input traces. 
\begin{align}
A&=V\times V\\
B&=V\times V.\bigcup_{i=1}^{m-n+1} \Sigma_I^{i}\\
C&=\{(\alpha,\beta)~|~ \alpha\in \text{pref}(\beta), \alpha, \beta\in V.\bigcup_{i=1}^{m-n+1} \Sigma_I^{i}\}
\end{align}
For any $\alpha,\beta\in A\cup B\cup C$, define 
\begin{align}
\Delta_M(\alpha,\beta)&=\{\gamma\in \Sigma_I^*~|~\omega_M(\ii q\after\alpha,\gamma) \neq
\omega_M(\ii q\after\beta,\gamma)\}\\
\Delta_{M_1}(\alpha,\beta)&=\{\gamma\in \Sigma_I^*~|~\omega_{M_1}(\ii q\after\alpha,\gamma) \neq
\omega_{M_1}(\ii q\after\beta,\gamma)\}\\
\end{align}
The set $\Delta_M(\alpha,\beta)$ contains all input traces distinguishing the states
$q\after\alpha$ and $q\after\beta$ in $M$ by yielding different output traces when applied to these
states. Set $\Delta_{M_1}(\alpha,\beta)$ contains all input traces 
distinguishing the states
$q\after\alpha$ and $q\after\beta$ in $M_1$.
Note that $\Delta_M(\alpha,\beta)$ and $\Delta_{M_1}(\alpha,\beta)$
may be empty, since $\ii q\after \alpha$ and 
$\ii q\after \beta$ are not necessarily distinguishable in $M$ or $M_1$, respectively.
For an arbitrary set $P\subseteq \Sigma_I^*\times \Sigma_I^*$ of pairs of input traces,  
define $P(M)\subseteq P$ and $P(M_1)\subseteq P$ by 
\begin{align}
(\alpha,\beta)\in P(M)&\Leftrightarrow (\alpha,\beta)\in P\wedge\Delta_M(\alpha,\beta)\neq \varnothing\\
(\alpha,\beta)\in P(M_1)&\Leftrightarrow (\alpha,\beta)\in P\wedge\Delta_{M_1}(\alpha,\beta)\neq \varnothing
\end{align}
We will apply this notation to the sets $A, B, C$ defined above: $A(M)$, for example, is the 
subset of all trace pairs $(\alpha,\beta)$ 
from $A$ whose target states $\ii q\after\alpha$ and $\ii q\after\beta$, respectively,
  are distinguishable in $M$.

\subsection{Main Theorem on Exhaustive Test Suites}
The following main theorem shows that the criteria \eqref{eq:tsone} and \eqref{eq:tstwo} 
suffice to guarantee that the test suite $\TS_\Rightarrow$ is exhaustive.

\begin{theorem}\label{th:se}
Let $m\ge n$ be a positive integer. Let $S=(S, \ii s, \Sigma_I,\Sigma_O, h_s)$ be a minimal DFSM with $|S|\le m$. Let $\TS_\Rightarrow\subseteq \Sigma_I^*$  be any test suite satisfying 
\begin{align}
&V.\bigcup_{i=0}^{m-n+1}\Sigma_I^i\subseteq \TS_\Rightarrow,\, \text{and}\label{eq:tsone}\\
&\forall (\alpha,\beta)\in A(M)\cup B(M_1)\cup C(M_1): \exists \gamma\in \Delta_M(\alpha,\beta): \alpha.\gamma, \beta.\gamma\in \TS_\Rightarrow
\label{eq:tstwo}
\end{align}
Then $\TS_\Rightarrow$   is exhaustive for composite requirement
 $R=R(q_1,x_1,Z_1)\wedge \dots \wedge R(q_k,x_k,Z_k)$, that is,
\[
S\ \pass_\Rightarrow\ \TS_\Rightarrow \Rightarrow S\models R.
\]
\end{theorem}
\begin{proof}
Suppose $S\ \pass_\Rightarrow\ \TS_\Rightarrow$.
We first show that  $V.\bigcup_{i=0}^{m-n} \Sigma_I^{i}$ is a state cover of $S$. Since $V$ is
a state cover of $M$ and $M$ is a prime machine, there are input traces 
$\alpha.\gamma, \beta.\gamma$ in $\TS_\Rightarrow$ for each pair of the $n$ states in $Q$, such that $\alpha,\beta\in V$ and states $\ii q\after \alpha$ and
$\ii q\after\beta$ are distinguished by $\gamma$ in $M$. Since $S$ passes these test cases, 
this also distinguishes $n$ states in $S$. Since $S$ has at most $m$ states, Lemma~\ref{lemma:statecov} can be applied to conclude that  $V.\bigcup_{i=0}^{m-n} \Sigma_I^{i}$ reaches all states of $S$.

Suppose $S\,\not\models\, R$. Then there is some $t\in\{ 1,\dots, k\}$ with $S\not\models R(q_t,x_t,Z_t)$. Hence there exists 
$\pi\in \Pi(q_t)$  with $\omega_M(\ii q\after\pi,x_t)\in Z_t$, but 
$\omega_S(\ii q\after\pi,x_t)\not\in Z_t$. Lifting these observations to the abstraction $M_2$
introduced in Section~\ref{sec:m2}, this induces the existence of input traces 
$\pi_2\in [\Pi(q_t)]$
such that $\omega_{M_2}([\ii q]\after\pi_2,x_t) = Z_t$ and 
$\omega_S(\ii s\after\pi_2,x_t)\not\in Z_t$.   
Since $\varepsilon \in V$, each   $\pi_2\in [ \Pi(q_t)]$ 
can be structured as $\pi_2 = v.\tau$ with
$v\in V$ and $\tau\in \Sigma_I^*$.
Let $\tau$ be a shortest sequence such that 
\begin{equation} 
\exists v\in V: v.\tau\in [ \Pi(q_t)]   \wedge \omega_S(\ii s\after v.\tau,x_t)\not\in Z_t.\label{eq:shorti2}
\end{equation}
 Since $S\ \pass_\Rightarrow\ \TS_\Rightarrow$ and  $V.\bigcup_{i=0}^{m-n+1} \Sigma_I^{i}\subseteq \TS_\Rightarrow$, we have 
 $|\tau|\ge m-n+1$, because otherwise, the input trace $v.\tau.x_t$ would have been tested, and the   test suite would have failed. 
 
Let $\tau_i=\tau^{[1..i]}, i\le m-n+1$. Then $v.\tau_i\neq v.\tau_j$ for all $i\neq j$. Let $V=\{v_1,\dots, v_n\}$. Then $v_i\neq v_j$ for $i\neq j$, since $V$ reaches $n$ different 
states in $M$. 
Suppose $v_i=v.\tau_j$ for some $i,j$. Let $\iota$ be the suffix of $\tau_j$ with $\tau_j.\iota=\tau$. Then $v_i.\iota=v.\tau_j.\iota=v.\tau\in [ \Pi(q_t)]$,  and $|\iota|<|\tau|$, a contradiction to the assumption that $\tau\in \Sigma_I^*$ is a shortest sequence satisfying condition~\eqref{eq:shorti2}. 
As a consequence, the set 
$U = \{v_1,\dots, v_n, v.\tau_1, \dots, v.\tau_{m-n+1}\}$ contains  $m+1$ elements. Observe that
$U\subseteq \TS_\Rightarrow$, since all elements start with a $v\in V$ and -- if extended by some $\tau_i$ -- are followed by an input trace of length less or equal to $m-n+1$.

Since $S$ contains only $m$ states, 
there exist $\alpha\neq \beta\in  U$ reaching the same state in $S$, that is, 
$\ii s\after\alpha=\ii s\after\beta$. Assume that 
$\ii q\after\alpha, \ii q\after\beta$ are {\it not} equivalent in $M_1$. This would
imply $\{\alpha,\beta\}\subseteq A(M) \cup B(M_1)\cup C(M_1)$ and the existence of a 
distinguishing trace $\gamma\in\Delta_M(\alpha,\beta)$ such that 
$\{ \alpha.\gamma, \beta.\gamma \} \subseteq \TS_\Rightarrow$ (see \eqref{eq:tstwo}). 
This would lead to a failed test execution,
since $\gamma$ cannot distinguish  $\ii s\after\alpha=\ii s\after\beta$, but would lead
to different outputs when applied to $\ii q\after\alpha$ and $\ii q\after\beta$. This contradiction implies that $\ii q\after\alpha, \ii q\after\beta$ are   equivalent in $M_1$.
As a consequence, $\{ \alpha, \beta \}$ cannot be contained in $V$, since pairs from $V$ 
always reach distinguishable states in $M$.  Without loss of generality, we therefore 
assume that one of the cases (a) $\alpha=v_i\in V, \beta=v.\tau_j$,
 or (b)~$\alpha=v.\tau_i, \beta=v.\tau_j, i<j$ applies.

 Let $\iota$ be the suffix of $\tau_j$ with $\tau_j.\iota=\tau$ 
 (in the case $j=m-n+1=|\tau|, \iota=\varepsilon$). Then $\beta.\iota=v.\tau$. Since $\ul q\after\alpha, \ul q\after\beta$ are equivalent states in $M_1$, we obtain that $\ul q\after\beta.\iota$ and  $\ul q\after\alpha.\iota$ are equivalent states in $M_1$,
 hence $\alpha.\iota\in [ \Pi(q_t)]$. Since $\ii s\after \beta=\ii s\after \alpha$ in $S$, 
 we conclude
 that $\ii s\after \beta.\iota=\ii s\after \alpha.\iota$ and $ \omega_S(\ii s\after \beta.\iota,x_t)=\omega_S(\ii s\after \alpha.\iota, x_t)\not \in Z_t$. In the case $\alpha=v_i\in V$, we have $|\iota|<|\tau|$. For $\alpha=v.\tau_i$, we have $\alpha.\iota=v.\tau_i.\iota$ and $|\tau_i.\iota|<|\tau|$. Both cases contradict the assumption that 
 $\tau\in \Sigma_I^*$ is a shortest sequence satisfying condition~\eqref{eq:shorti2}. 
 This contradiction implies that $S\ \pass_\Rightarrow\ \TS_\Rightarrow$ and
 $S \not\models R$  cannot both be true and completes the proof of the theorem.
\xbox
\end{proof}

Algorithms for creating test suites according to formulas \eqref{eq:tsone} and \eqref{eq:tstwo} have been described in the original work~\cite{DBLP:conf/forte/DorofeevaEY05} and in~\cite{Huang2018}.

\begin{example}
\label{ex:exhaustive}
Consider again reference model $M$ and the (in black box testing practise unknown) implementation model $S$ from Example~\ref{ex:runninga}. $M$ is minimal and has $n=3$ states. Under the hypothesis that the minimised model of the true implementation behaviour also has $m=3$ states, the following test suite with 4 test cases has been 
calculated according to the rules \eqref{eq:tsone} and \eqref{eq:tstwo}.
\[
a.a.b, a.b.b, b.a.b, b.b.a
\]
Applying this test suite to the implementation model $S$ results in I/O-traces
\[
a.a.b/1.0.2,
a.b.b/1.2.0,
b.a.b/2.0.0,
b.b.a/2.0.1
\]
These executions conform to the reference model $M$, so the implementation passes the suite. 
This shows that $S \models R$ holds.

Applying the complete H-method to $M$ and hypothesis $m=n=3$  results in 5 test cases.
These comprise the ones calculated according to \eqref{eq:tsone} and \eqref{eq:tstwo}, adding a fifth test case $a.a.a$.
This new test case  reveals the
fact that $S$ is not language equivalent to $M$, since $S$ produces $a.a.a/1.0.1$ where $a.a.a/1.0.0$ had been expected 
according to reference model $M$.

In Appendix~\ref{sec:fsmlib}, it is explained how these test suites can be automatically generated using the 
library fsmlib-cpp mentioned in the introduction.
\qed
\end{example}

\section{Evaluation Using a Real-World Example}
\label{sec:rwexample}

\subsection{Fasten Seatbelt and Return-to-Seat Sign Control}

The following experiment is a (slightly simplified) real-world example concerning safety-related and uncritical indications in an aircraft cabin\footnote{The application used in this experiment has been originally published in~\cite[Section~4.1]{Huang2018}. The application description has been reproduced here in its original form, in order to make this article sufficiently self-contained.}.

\begin{table}
\caption{State-transition table of DFSM specifying the control of FSB signs and RTS signs in an aircraft cabin.}
\begin{center}
\begin{tabular}{|l||c|c|c|c|c|c|c|c|c|}\hline\hline
  & {\bf f0} & {\bf f1} & {\bf f2} &{\bf d1} &{\bf d0} & {\bf e1} & {\bf e0} & {\bf a1} & {\bf a0} \\\hline\hline
$\mathbf{s_0}$ & $s_0$/00 & $s_1$/11 & $s_2$/00 & $s_3$/10 & $s_0$/00 & $s_6$/10 & $s_0$/00 & $s_{12}$/00 & $s_0$/00  \\
$\mathbf{s_1}$ & $s_0$/00 & $s_1$/11 & $s_2$/00 & $s_4$/10 & $s_1$/11 & $s_7$/10 & $s_1$/11 & $s_{13}$/11 & $s_1$/11  \\
$\mathbf{s_2}$ & $s_0$/00 & $s_1$/11 & $s_2$/00 & $s_5$/10 & $s_2$/00 & $s_8$/10 & $s_2$/00 & $s_{14}$/11 & $s_2$/00  \\
$\mathbf{s_3}$ & $s_3$/10 & $s_4$/10 & $s_5$/10 & $s_3$/10 & $s_0$/00 & $s_9$/10 & $s_3$/10 & $s_{15}$/10 & $s_3$/10  \\
$\mathbf{s_4}$ & $s_3$/10 & $s_4$/10 & $s_5$/10 & $s_4$/10 & $s_1$/11 & $s_{11}$/10 & $s_4$/10 & $s_{16}$/10 & $s_4$/10  \\
$\mathbf{s_5}$ & $s_3$/10 & $s_4$/10 & $s_5$/10 & $s_5$/10 & $s_2$/00 & $s_{11}$/10 & $s_5$/10 & $s_{17}$/10 & $s_5$/10  \\
$\mathbf{s_6}$ & $s_6$/10 & $s_7$/10 & $s_8$/10 & $s_9$/10 & $s_6$/10 & $s_6$/10 & $s_0$/00 & $s_{18}$/10 & $s_6$/10  \\
$\mathbf{s_7}$ & $s_6$/10 & $s_7$/10 & $s_8$/10 & $s_{10}$/10 & $s_7$/10 & $s_7$/10 & $s_1$/11 & $s_{19}$/10 & $s_7$/10  \\
$\mathbf{s_8}$ & $s_6$/10 & $s_7$/10 & $s_8$/10 & $s_{11}$/10 & $s_8$/10 & $s_8$/10 & $s_2$/00 & $s_{20}$/10 & $s_8$/10  \\
$\mathbf{s_9}$ & $s_9$/10 & $s_{10}$/10 & $s_{11}$/10 & $s_9$/10 & $s_6$/10 & $s_9$/10 & $s_3$/10 & $s_{21}$/10 & $s_9$/10  \\
$\mathbf{s_{10}}$ & $s_9$/10 & $s_{10}$/10 & $s_{11}$/10 & $s_{10}$/10 & $s_7$/10 & $s_{10}$/10 & $s_4$/10 & $s_{22}$/10 & $s_{10}$/10  \\
$\mathbf{s_{11}}$ & $s_9$/10 & $s_{10}$/10 & $s_{11}$/10 & $s_{11}$/10 & $s_8$/10 & $s_{11}$/10 & $s_5$/10 & $s_{23}$/10 & $s_{11}$/10  \\
$\mathbf{s_{12}}$ & $s_{12}$/00 & $s_{13}$/11 & $s_{14}$/11 & $s_{15}$/10 & $s_{12}$/00 & $s_{18}$/10 & $s_{12}$/00 & $s_{12}$/00 & $s_0$/00  \\
$\mathbf{s_{13}}$ & $s_{12}$/00 & $s_{13}$/11 & $s_{14}$/11 & $s_{16}$/10 & $s_{13}$/11 & $s_{19}$/10 & $s_{13}$/11 & $s_{13}$/11 & $s_1$/11  \\
$\mathbf{s_{14}}$ & $s_{12}$/00 & $s_{13}$/11 & $s_{14}$/11 & $s_{17}$/10 & $s_{14}$/11 & $s_{20}$/10 & $s_{14}$/11 & $s_{14}$/11 & $s_2$/00  \\
$\mathbf{s_{15}}$ & $s_{15}$/10 & $s_{16}$/10 & $s_{17}$/10 & $s_{15}$/10 & $s_{12}$/00 & $s_{21}$/10 & $s_{15}$/10 & $s_{15}$/10 & $s_3$/10  \\
$\mathbf{s_{16}}$ & $s_{15}$/10 & $s_{16}$/10 & $s_{17}$/10 & $s_{16}$/10 & $s_{13}$/11 & $s_{22}$/10 & $s_{16}$/10 & $s_{16}$/10 & $s_4$/10  \\
$\mathbf{s_{17}}$ & $s_{15}$/10 & $s_{16}$/10 & $s_{17}$/10 & $s_{17}$/10 & $s_{14}$/11 & $s_{23}$/10 & $s_{17}$/10 & $s_{17}$/10 & $s_5$/10  \\
$\mathbf{s_{18}}$ & $s_{18}$/10 & $s_{19}$/10 & $s_{20}$/10 & $s_{21}$/10 & $s_{18}$/10 & $s_{18}$/10 & $s_{12}$/00 & $s_{18}$/10 & $s_6$/10  \\
$\mathbf{s_{19}}$ & $s_{18}$/10 & $s_{19}$/10 & $s_{20}$/10 & $s_{22}$/10 & $s_{19}$/10 & $s_{19}$/10 & $s_{13}$/11 & $s_{19}$/10 & $s_7$/10  \\
$\mathbf{s_{20}}$ & $s_{18}$/10 & $s_{19}$/10 & $s_{20}$/10 & $s_{23}$/10 & $s_{20}$/10 & $s_{20}$/10 & $s_{14}$/11 & $s_{20}$/10 & $s_8$/10  \\
$\mathbf{s_{21}}$ & $s_{21}$/10 & $s_{22}$/10 & $s_{23}$/10 & $s_{21}$/10 & $s_{18}$/10 & $s_{21}$/10 & $s_{15}$/10 & $s_{21}$/10 & $s_9$/10  \\
$\mathbf{s_{22}}$ & $s_{21}$/10 & $s_{22}$/10 & $s_{23}$/10 & $s_{22}$/10 & $s_{19}$/10 & $s_{22}$/10 & $s_{16}$/10 & $s_{22}$/10 & $s_{10}$/10  \\
$\mathbf{s_{23}}$ & $s_{21}$/10 & $s_{22}$/10 & $s_{23}$/10 & $s_{23}$/10 & $s_{20}$/10 & $s_{23}$/10 & $s_{17}$/10 & $s_{23}$/10 & $s_{11}$/10  \\
\hline\hline
\end{tabular}
\end{center}

\hspace*{12mm}
\begin{minipage}[t]{16cm}
First column defines  the states (initial state $s_0$)
\newline 
First row defines the inputs
\newline
Fields $s/y$ denote `Post-state/Output'
\\

Inputs:
\\
{\bf f0, f1, f2} : FSB switch in position OFF, ON, AUTO 
\\
{\bf d1, d0} : Cabin decompression true, false  
\\
{\bf e1, e0} : Excessive altitude true, false 
\\
{\bf a1, a0} : Auto condition true, false 
\\
Outputs:
\\
00 denotes (FSB,RTS)=(0,0)
\\
11 denotes (FSB,RTS)=(1,1)
\\
10 denotes (FSB,RTS)=(1,0)
\end{minipage}
\label{tab:fsb}
\end{table}%

\begin{table}
\caption{Explanation of states $\mathbf{s_0}$,\dots, $\mathbf{s_{23}}$ in Table~\ref{tab:fsb}.}
\begin{center}
\begin{tabular}{|l||p{15mm}|p{22mm}|p{15mm}|p{20mm}|p{15mm}|}\hline\hline
  & {\bf Cockpit Switch} & {\bf Decompression} & {\bf Excessive\newline Altitude} &{\bf AUTO\newline Condition} & {\bf Current\newline Output} \\\hline\hline
$\mathbf{s_0}$ &0 & 0 & 0 & 0  & 00  \\
$\mathbf{s_1}$ &1 & 0 & 0 & 0  & 11  \\
$\mathbf{s_2}$ &2 & 0 & 0 & 0  & 00  \\
$\mathbf{s_3}$ &0 & 1 & 0 & 0  & 10  \\
$\mathbf{s_4}$ &1 & 1 & 0 & 0  & 10  \\
$\mathbf{s_5}$ &2 & 1 & 0 & 0  & 10  \\
$\mathbf{s_6}$ &0 & 0 & 1 & 0  & 10  \\
$\mathbf{s_7}$ &1 & 0 & 1 & 0  & 10  \\
$\mathbf{s_8}$ &2 & 0 & 1 & 0  & 10  \\
$\mathbf{s_9}$ &0 & 1 & 1 & 0  & 10  \\
$\mathbf{s_{10}}$ &1 & 1 & 1 & 0  & 10  \\
$\mathbf{s_{11}}$ &2 & 1 & 1 & 0  & 10  \\
$\mathbf{s_{12}}$&0 & 0 & 0 & 1  & 00  \\
$\mathbf{s_{13}}$ &1 & 0 & 0 & 1  & 11  \\
$\mathbf{s_{14}}$ &2 & 0 & 0 & 1  & 11  \\
$\mathbf{s_{15}}$ &0 & 1 & 0 & 1  & 10  \\
$\mathbf{s_{16}}$ &1 & 1 & 0 & 1  & 10  \\
$\mathbf{s_{17}}$ &2 & 1 & 0 & 1  & 10  \\
$\mathbf{s_{18}}$ &0 & 0 & 1 & 1  & 10  \\
$\mathbf{s_{19}}$ &1 & 0 & 1 & 1  & 10  \\
$\mathbf{s_{20}}$ &2 & 0 & 1 & 1  & 10  \\
$\mathbf{s_{21}}$ &0 & 1 & 1 & 1  & 10  \\
$\mathbf{s_{22}}$&1 & 1 & 1 & 1  & 10  \\
$\mathbf{s_{23}}$ &2 & 1 & 1 & 1  & 10  \\
\hline\hline
\end{tabular}
\end{center}
\label{tab:fsbstates}
\end{table}%

A \emph{cabin controller} in a
modern aircraft switches the \emph{fasten seat belt (FSB) signs} located 
above the passenger seats in the cabin and the \emph{return to seat (RTS) signs}
located in the lavatories according to the rules modelled in the DFSM
shown in Table~\ref{tab:fsb}.  Note that this DFSM is already minimal.

As inputs, the cabin controller reads the actual position of the fasten seat belts switch in the cockpit, which has the position {\bf f0} (OFF), {\bf f1} (ON), and {\bf f2} (AUTO). Further inputs come from the cabin pressure control system which indicates
``cabin pressure low'' by event {\bf d1} and ``cabin pressure ok'' by {\bf d0}. 
This controller also indicates ``excessive altitude'' by {\bf e1} or
``altitude in admissible range'' by {\bf e0}. Another sub-component of the cabin controller determines whether the so-called AUTO condition is true (event {\bf a1}) or false ({\bf a0}). 

The cabin controller switches the fasten seat belt signs and return to seat signs
on and off, depending on the actual input change and its current internal state.
As long as the cabin pressure and the cruising altitude are ok (after initialisation of the cabin controller or if last events from the cabin pressure controller were  {\bf d0, e0}), the status of the FSB and RTS signs is determined by the cockpit switch and the AUTO condition: if the switch is in the ON position, both FSB and RTS signs are switched on (output 11 in Table~\ref{tab:fsb}). Turning the switch into the OFF position switches the signs off. If the switch is in the AUTO position, both FSB and RTS signs are switched on if the AUTO condition becomes true with event {\bf a1}, and they are switched off 
again after event {\bf a0}. The AUTO condition may depend on the status of landing gears, slats, flaps, and oil pressure, these details are abstracted to {\bf a1, a0} in 
our example.

As soon as    a loss of pressure occurs in the cabin (event {\bf d1}) or an excessive altitude is reached, the FSB signs must be switched on and remain in this state, regardless of the actual state of the cockpit switch and the AUTO condition. The RTS signs, however, need to be switched off, because passengers should not be encouraged to leave the lavatories in a low pressure or excessive altitude situation. 

After the cabin pressure and the altitude are back in the admissible range, the FSB and RTS signs shall automatically resume their state as determined by the ``normal'' inputs from cockpit switch and AUTO condition.

Table~\ref{tab:fsbstates} facilitates the interpretation of the DFSM model shown in Table~\ref{tab:fsb}: for every DFSM 
state $\mathbf{s_0}$,\dots,$\mathbf{s_{23}}$, the associated status of the cockpit switch, cabin decompression, excessive altitude, and AUTO condition, as well
as the last output made when entering the state
is displayed.

\subsection{Complete Test Suites Checking Model Equivalence Computed by the H-Method}

Applying the H-Method~\cite{DBLP:conf/forte/DorofeevaEY05} as implemented in the fsmlib-cpp\footnote{see Appendix~\ref{sec:fsmlib} for instructions how to use the test case generation program which is part of the library}, the number of test cases needed to test
an implementation to establish language equivalence with full fault coverage is shown in Table~\ref{tab:hversusimplication}, column {\bf H}. 
Recall that a test case is a sequence of inputs. When executing a test case against an implementation, the expected 
results are determined 
As the size of the test suite depends on the potential number $m$ of implementation states minus the number $n$ of states in the minimised model, test  suites are calculated for $m-n = 0,1,2$.

\begin{table}[htp]
\caption{Comparison of the numbers of test cases needed to prove I/O-equivalence (H-Method) and to prove requirements satisfaction using exhaustive   testing for requirements $\mathbf{R}_1$ and $\mathbf{R}_2$ specified in this section.}
\begin{center}
\begin{tabular}{|r||r|r||r|r|r||r|r|r||}\hline\hline
$\mathbf{m-n}$ & {\bf H} &  $t_H[s]$   &$\mathbf{R}_1$ & $\mathbf{\Delta_1^\%}$ & $t_1[s]$ & 
$\mathbf{R}_2$  &  $\mathbf{\Delta_2^\%}$ &  $t_2[s]$
\\\hline\hline
 0 & 518 & 1.3 & 193 & 63 & 0.02 & 337 & 35 & 0.03
 \\\hline
 1 & 4069 & 76.7 & 1737 & 57 & 0.09 & 3035 & 25 & 0.2
 \\\hline
 2 & 35325 & 7765.6& 15633 & 56 & 0.9 & 27327 & 23 & 2.3
\\\hline\hline
\end{tabular}
\end{center}
\label{tab:hversusimplication}
Column $\mathbf{m-n}$ contains the maximal difference between the number of SUT states and model states assumed for the test suite generation. 
Column {\bf H} contains the number of test cases required
when using the H-Method for I/O-equivalence testing. Columns $\mathbf{R}_i,\ i=1,2$ 
contain the number
of test cases required for requirements-driven testing of $\mathbf{R}_i$.
Columns $\mathbf{\Delta}_i^\%,\ i=1,2$ contain 
the test case reductions achieved by requirements-driven testing in percent, calculated 
according to formulas $\mathbf{\Delta}_i^\% = 100 - 100*\frac{R_i}{H}$, where $R_i$ denotes the number 
of test cases from column $\mathbf{R}_i$  and $H$ the number from column {\bf H}. Time $t_H$ is the test suite generation time in seconds needed to create the respective language equivalence test suite using the H-Method. Durations $t_i,\ i=1,2$
indicate the time needed to generated the exhaustive requirements test suites for $\mathbf{R}_1$ and
$\mathbf{R}_2$, respectively. The time has been measured on an Apple iMac with 4.2~GHz Intel Core~i7 CPU and 64GB 2400 MHz DDR4 memory.
\end{table}%

\subsection{Requirement $\mathbf{R}_1$: Safety-relevant outputs on Decompression}

As a first requirement, we consider
\begin{quote}
$\mathbf{R}_1$. \it Whenever cabin decompression occurs, the FSB signs shall be set to 1, and the RTS signs to 0.
\end{quote}
In the DFSM model shown in Table~\ref{tab:fsb}, this requirement is reflected by column $\mathbf{d1}$: regardless of 
the current state, output (FSB,RTS)=(1,0) is assigned on occurrence of input $\mathbf{d1}$. Encoding $\mathbf{R}_1$
in our requirements specification formalism results in the following representation.
$$
\mathbf{R}_1 \equiv \bigwedge_{i=0}^{23} R(s_i,\mathbf{d1},\{ 10 \})
$$
Obviously, the specification of $\mathbf{R}_1$ is independent of the current state, and  there are no alternative outputs that may be accepted in exchange for the output  (FSB,RTS)=(1,0) expected according to the reference model.

Creating an exhaustive test suite for $\mathbf{R}_1$ according to Theorem~\ref{th:se} results in the number
of test cases shown in column  $\mathbf{R}_1$ of Table~\ref{tab:hversusimplication}. As listed in column 
$\mathbf{\Delta_1^\%}$ of this table, the test case reductions achieved in comparison to the H-based test suite   vary with $m-n$ and are in range $56\%$ --- $63\%$.

The test of this requirement could be advisable, for example, during regression testing after the sign controller's 
code  has been 
modified with regard to the sub-function responsible for decompression handling. Instead of running the complete H-test suite, it would suffice to perform the significantly smaller test suite for  $\mathbf{R}_1$ which also guarantees 
full fault coverage as far as  $\mathbf{R}_1$ is concerned.

\subsection{Requirement $\mathbf{R}_2$: All Safety-relevant outputs and Sign Activation by Manual Switch}

As second requirement, we consider a composite requirement related to all safety-critical events in combination 
to normal behaviour reactions (FSB,RTS)=(1,1) when setting the cockpit switch to position 1.
\begin{quote}
$\mathbf{R}_2$. \it Whenever cabin decompression or excessive altitude occurs, the FSB signs shall be set to 1, and the RTS signs to 0. The signs stay activated until both decompression and execessive altitude are no longer present. In absence of 
cabin decompression and excessive altitude, both FSB and RTS signs shall be switched on when setting the cockpit switch to position 1.
\end{quote}
Formalising this requirement leads to
\begin{eqnarray*}
\mathbf{R}_2 & \equiv & \bigwedge_{i=0}^{23} \big( R(s_i,\mathbf{d1},\{ 10 \}) 
                                                 \wedge  R(s_i,\mathbf{e1},\{ 10 \}) \big) \wedge {}
\\
& & \bigwedge_{i\in\{3,\dots,11,15,\dots,23  \}} 
\big( R(s_i,\mathbf{f0},\{ 10 \}) \wedge R(s_i,\mathbf{f1},\{ 10 \}) \wedge  R(s_i,\mathbf{f2},\{ 10 \})  \big) \wedge {}
\\
& & \bigwedge_{i\in\{3,\dots,11,15,\dots,23  \}} 
\big( R(s_i,\mathbf{a1},\{ 10 \}) \wedge R(s_i,\mathbf{a0},\{ 10 \}) \big) \wedge {}
\\
& & \bigwedge_{i\in\{6,\dots,11,18,\dots,23  \}} R(s_i,\mathbf{d0},\{ 10 \}) \wedge {}
\\
& & \bigwedge_{i\in\{ 3,4,5,9,10,11,15,16,17,21,22,23  \}} R(s_i,\mathbf{e0},\{ 10 \})
\\
& & \bigwedge_{i\in\{0,1,2,12,13,14  \}} R(s_i,\mathbf{f1},\{ 11 \})
\end{eqnarray*}

Creating an  exhaustive test suite for $\mathbf{R}_2$ according to Theorem~\ref{th:se} results in the number
of test cases shown in column  $\mathbf{R}_2$ of Table~\ref{tab:hversusimplication}. As listed in column 
$\mathbf{\Delta_2^\%}$ of this table, the test case reductions achieved in comparison to the H-based test suite    are in range $23\%$ --- $35\%$.

Testing this requirement would be advisable, for example, after the initial development of the sign controller, assuming that there would not be enough time to perform all tests derived by the H-method. Requirement $\mathbf{R}_2$ leads 
to a smaller test suite, but still covers all safety-relevant reactions of the implementation and the most important
user requirement. Then it could be justified to test the   functionality related to the AUTO condition with less effort, since, as a fall back option for the pilot, signs could always be switched on manually when the AUTO mode is not properly functioning.

\section{Complete Testing of Composite Requirements}
\label{sec:maintheorem}

Suppose that DFSMs $M$ (reference model) and $S$ (implementation) 
are completely specified and consider the composite requirement  $R=\bigwedge_{i=1}^k R(q_i,x_i,Z_i)$   specified on $M$.
In this section, we are going to construct finite test suites $\TS_\Leftrightarrow$ depending on $M$ and $R$ and a
new pass relation $\pass_\Leftrightarrow$ such that an implementation $S$  passes 
$\TS_\Leftrightarrow$ {\it if and only if} it conforms to the requirements $R$. In the terminology
introduced in Section~\ref{sec:maincontrib}, these  test suites are called complete.

\subsection{A Nondeterministic Alternative to $M_1$}\label{sec:m1prime}

Let $M_1$ be the FSM abstraction induced by $M$ and $R$ as described in Section~\ref{sec:dfsmabs}. 
Recall that $M_1$ is deterministic with output alphabet $\{*, Z_1,\dots, Z_k\}$, where
$*$ is short for $\Sigma_O$. 
The DFSM $M_1$   induces an alternative abstraction $M_1'$  of $M$, which is
{\it nondeterministic} and has   output alphabet $\Sigma_O$. This nondeterministic FSM is specified by 
$M_1'=(Q, \ul q, \Sigma_I, \Sigma_O, h_1')$, where 
$$
(q,x,y,q')\in h_1'\Leftrightarrow \big(\delta_{M_1}(q,x) = q' \wedge y \in \omega_{M_1}(q,x)\big).
$$
Given $(q,x)\in \Sigma_I\times \Sigma_I$ such that $\omega_{M_1}(q,x) = *$,
the nondeterministic machine $M_1'$  possesses the transitions $(q,x,y,\delta_{M_1}(q,x))$ with arbitrary
$y\in\Sigma_O$.
For $i\in\{1,\dots,k \}$, recall that $\omega_{M_1}(q_i,x_i) = Z_i$. For these cases, 
$M_1'$  reacts by nondeterministically taking one of the
transitions $(q_i,x_i,y,\delta_{M_1}(q,x))$ with    $y\in Z_i$.

\begin{example}\label{ex:fsm_abs_nondet}
When representing the abstraction $M_1$ from Fig.~\ref{fig:abstractedM1} as a nondeterministic
FSM, this results in $M_1'$ as depicted in Fig.~\ref{fig:fsm_abs_nondet}.
\qed
\end{example}

\begin{figure}[h]
\begin{center}
 \includegraphics[width=0.7\textwidth]{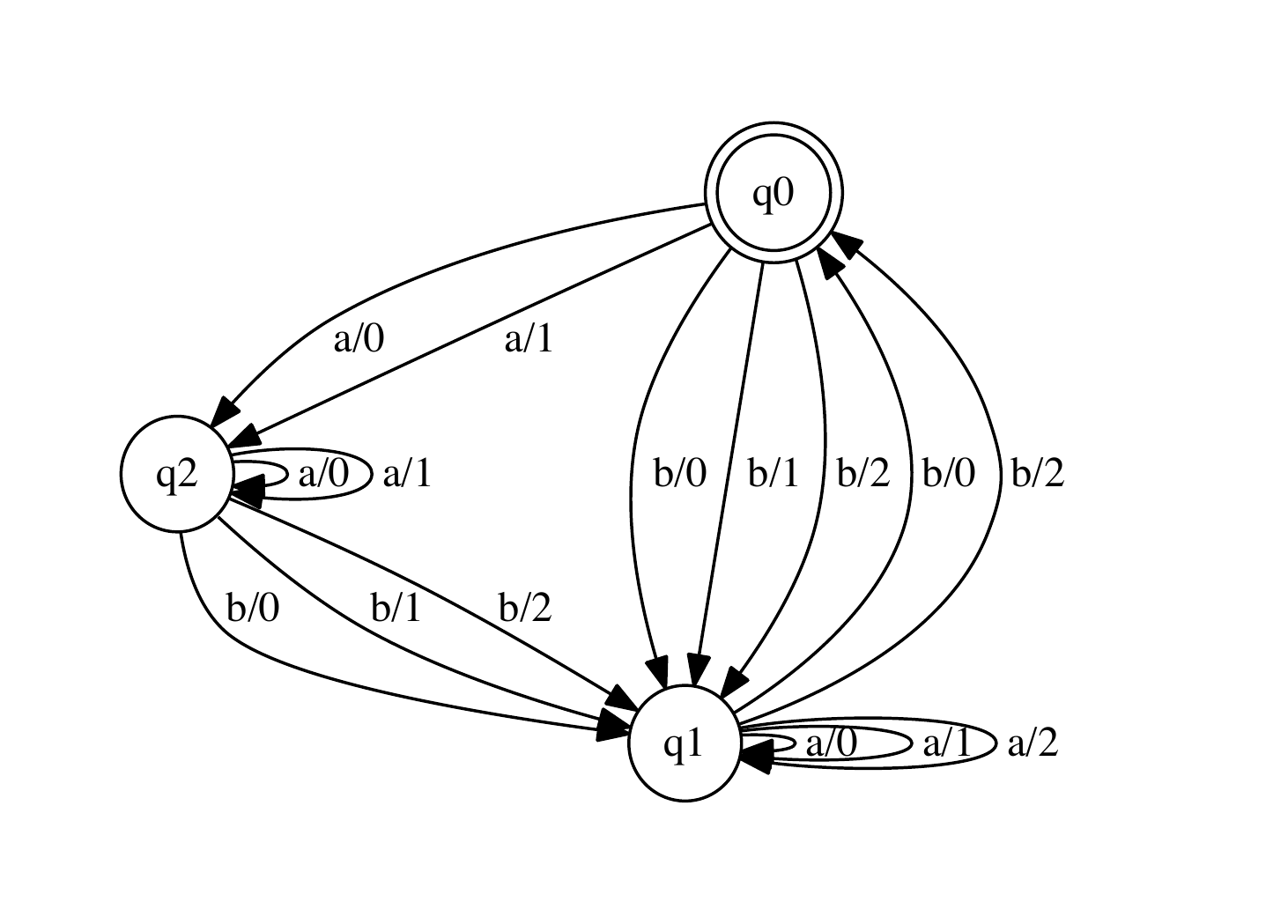}
\caption{Nondeterministic FSM $M_1'$ created from $M_1$ displayed in Fig.~\ref{fig:abstractedM1}.}
\label{fig:fsm_abs_nondet}
\end{center}
\end{figure}

From the construction rules specified above and from the illustration in Example~\ref{ex:fsm_abs_nondet} and Fig.~\ref{fig:fsm_abs_nondet}, it is immediately clear 
that  $M_1'$ is observable and completely specified. Also note that by construction, 
$M_1'$ may be interpreted as the {\it ``most nondeterministic FSM which still satisfies requirement $R$''}: for (state,input)-pairs that are unrelated to $R$, any output from $\Sigma_O$ can occur. For   pairs $(q_i,x_i)$ related to elementary requirements $R(q_i,x_i,Z_i)$, any
output from $Z_i$ can be nondeterministically selected, so that $R(q_i,x_i,Z_i)$ is never
violated.

Moreover, $M_1'$ is only nondeterministic with respect to 
the outputs produced in a given state $q$ for a 
given input $x$, whereas the target state is uniquely determined by $q$ and $x$. This means that 
the transition functions $\delta_{M_1}, \delta_{M_1'}: Q\times \Sigma_I \fun Q$ of $M_1$ and 
$M_1'$, respectively, coincide in the sense that 
$$
\forall (q,x)\in Q\times \Sigma_I: \delta_{M_1}(q,x) = \delta_{M_1'}(q,x).
$$
As a consequence, the expressions $\ii q\after \ol x,\ \ol x\in \Sigma_I^*,$ always 
result in the same uniquely determined 
target state, regardless of whether they are evaluated in $M_1$ or $M_1'$.\footnote{Recall that for general nondeterministic FSMs, $\ii q\after \ol x$ specifies a {\it set} of possible target states, since their transition 
relation $h$ may allow for different target states being reached for a given pre-state
and input.}

Finally, note that, since $M_1'$ is nondeterministic, its output function  is set-valued, 
$\omega_{M_1'}  :  Q\times \Sigma_I^* \fun \mathbb{P}(\Sigma_O^*)$. It is easy to see, however, that $\omega_{M_1'}(q,x) = \omega_{M_1}(q,x)$ for all states $q$ and inputs $x$:
FSM $M_1'$ has transitions
$q \xrightarrow{x/y} \delta_{M_1}(q,x)$ for every $y\in\omega_{M_1}(q,x)$.

The following theorem presents an important insight into the relationship between requirements satisfaction and reduction (i.e.~language inclusion): an implementation machine $S$ satisfies the requirement $R$, if and only if its language is contained in the language of the abstraction $M_1'$ constructed above. 
\begin{theorem}\label{th:reqreduct}
$S\models R\Leftrightarrow L(S)\subseteq L(M_1')$
\end{theorem}
\begin{proof}
Recall that $\Pi(q)=\{\ol x\in \Sigma_I^*~|~\ul q\after \ol x=q\}$. 
Also recall that by definition, $S\models R\Leftrightarrow \forall i\in\{1,\dots,k\},  \pi\in \Pi(q_i): \omega_S(\ii s\after \pi,x_i)\in Z_i$.

Suppose that $L(S)\subseteq L(M_1')$. Since $S$ is completely specified by assumption, 
there exists a unique 
I/O-trace $\ol x.x_i/\ol y.y \in L(S)$ for any input sequence $\ol x.x_i$
with $\ol x \in \Pi(q_i)$ and $i\in\{1,\dots,k\}$. Since we assume that $S$ is a reduction
of $M_1'$, this I/O-trace must also be a trace of $M_1'$. By construction of $M_1'$, this means 
that $y\in Z_i$, so $S$ fulfils requirement $R(q_i,x_i,Z_i)$. Since this argument was independent
of $i\in\{1,\dots,k\}$, $S\models R$ follows.  

Now suppose that $L(S)\not\subseteq L(M_1')$. Then there exists an I/O-trace 
$\ol x.x/\ol y.y \in L(S)$ such that $\ol x/\ol y \in L(M_1')$, but 
$\ol x.x/\ol y.y \not\in L(M_1')$. Suppose that $(\ii q\after \ol x,x) \neq (q_i,x_i)$ 
for all
$i\in\{1,\dots, k\}$, where the $\after$ operator is evaluated  in $M_1'$. Then, 
by construction of $M_1'$, $\ol x.x/\ol y.y \in L(M_1')$ for {\it all}
$y\in\Sigma_O$, so this is a contradiction to the assumption $\ol x.x/\ol y.y \not\in L(M_1')$.

Thus the assumption
$\ol x.x/\ol y.y \not\in L(M_1')$ implies the existence of an $i\in\{1,\dots, k\}$ such that 
$(\ii q\after \ol x,x) = (q_i,x_i)$. As a consequence,  
$\ol x.x_i/\ol y.y \not\in L(M_1')$ and $\ol x\in\Pi(q_i)$. By construction of $M_1'$, this means that $y\not\in Z_i$.
This implies that $S$ violates requirement $R(q_i,x_i,Z_i)$, so $S\not\models R$, and this 
completes the proof.
\qed
\end{proof}

\subsection{The pass criterion $\pass_\Leftrightarrow$}
Let $\ol x\in \Sigma_I^*$ and $\TS_\Leftrightarrow\subseteq \Sigma_i^*$. We define a new 
pass criterion for test cases by 
$$
S\ \pass_\Leftrightarrow\  \ol x \equiv \big( \omega_S(\ii s, \ol x)\in \omega_{M_1'}(\ul q, \ol x)\big),
$$
and extend this to test suites by
$$
S\ \pass_\Leftrightarrow\ \TS_\Leftrightarrow \equiv \big(\forall \ol x\in \TS_\Leftrightarrow: S\ \pass_\Leftrightarrow\  \ol x\big).
$$
Intuitively speaking, a test case represented by some input trace $\ol x$ is passed 
by the implementation $S$ if and only if the resulting I/O-trace $\ol x/\omega_S(\ii s, \ol x)$ 
performed by $S$ is contained in the language of $M_1'$.  Using the respective output functions, this is expressed here by stating that the output trace generated by $S$ on input trace $\ol x$
in an element of the set of output traces possible in $M_1'$ for this test case.
Therefore, $\pass_\Leftrightarrow$ is just the well-known 
pass criterion for reduction testing of $S$ against
$M_1'$, that is, for checking whether $L(S)\subseteq L(M_1')$ holds.


\begin{figure}[h]
\begin{center}
 \includegraphics[width=0.4\textwidth]{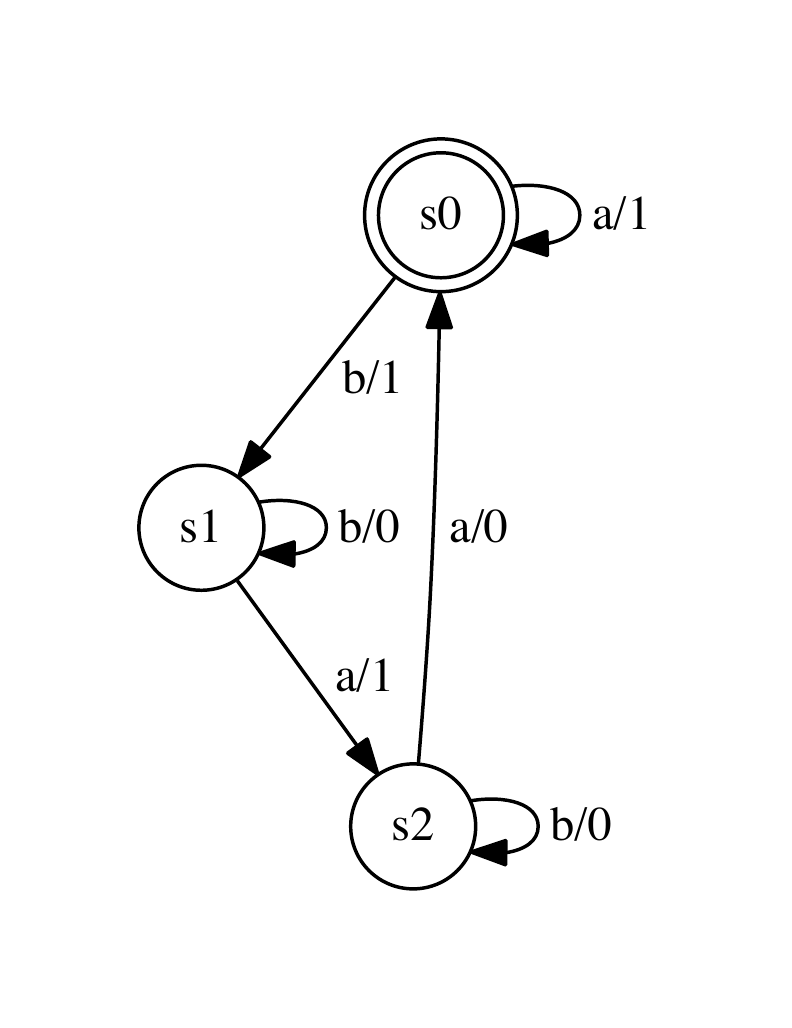}
\caption{Implementation $S$ from Example~\ref{ex:cannotuseexhaustive}.}
\label{fig:cannotuseexhaustive}
\end{center}
\end{figure}
\begin{example}\label{ex:cannotuseexhaustive}
At first glance, one might ask whether it is possible to apply the test suites specified 
in Theorem~\ref{th:se} just with the new pass criterion $\pass_\Leftrightarrow$, 
in order to create complete test suites for the given requirement. This, however, is not true: the following example shows that the suites from Theorem~\ref{th:se} are no longer exhaustive, when applied with $\pass_\Leftrightarrow$.

Consider again reference model $M$ and requirement 
\[ 
 R \equiv R(q_0,a,\{0,1\}) \wedge R(q_1,b,\{ 0,2\}) \wedge R(q_2,a,\{0,1\}).
 \]
 from Example~\ref{ex:runninga}. From Example~\ref{ex:exhaustive}, we know that 
 \[
a.a.b, a.b.b, b.a.b, b.b.a
\]
 is an exhaustive  test suite for pass criterion $\pass_\Rightarrow$ 
 generated according to rules \eqref{eq:tsone} and \eqref{eq:tstwo} from
 Theorem~\ref{th:se}.
Now consider another implementation $S$ as shown in Fig.~\ref{fig:cannotuseexhaustive}.
Applying the four test cases above to $S$ results in I/O-traces
\[
a.a.b/1.1.1,\ a.b.b/1.1.0,\ b.a.b/1.1.0,\ b.b.a/1.0.1
\]
It is easy to see that $S$ passes the four test cases when applying pass criterion 
$\pass_\Leftrightarrow$ (just check the observed outputs against $M_1$ from Fig.~\ref{fig:abstractedM1}). However, $S$ does not satisfy 
requirement $R(q_1,b,\{ 0,2\})$, and, equivalently, $S$ is not a reduction of $M_1'$
shown in Fig.~\ref{fig:fsm_abs_nondet}: input trace $b.a.a.b$ applied to $S$ results in 
\[
   b.a.a.b/1.1.0.1,
\]
but $\ii q\after b.a.a = q_1$ in  $M$, and $R(q_1,b,\{ 0,2\})$ only allows $0$ or $2$
as output when $b$ is applied in state $q_1$. In contrast to that, $S$ outputs 1 when given
input $b$ in state $s_0\after b.a.a = s_0$. Expressed in an equivalent way, 
$b.a.a.b/1.1.0.1\not\in L(M_1')$, for $M_1'$ shown in Fig.~\ref{fig:fsm_abs_nondet}. 
\xbox
\end{example}

\subsection{Main Theorem on Complete Test Suites}

The following theorem shows that any complete test suite for reduction testing against
$M_1'$ can be reduced in a specific way that still guarantees requirements satisfaction
if and only if the resulting suite is passed.

\begin{theorem}\label{th:sre}
Let $\TS$ be any complete reduction test suite guaranteeing $L(S)\subseteq L(M_1')$ if and only
if $S\ \pass_\Leftrightarrow \TS$ for all $S \in{\cal D}$.  Define 
$$
\ol \Pi=\bigcup_{i=1}^{k}\Pi(q_i).\{x_i\}.
$$
Then any test suite $\TS_\Leftrightarrow$ satisfying
$$
 \text{pref}(TS) \cap \overline \Pi \subseteq \TS_\Leftrightarrow 
$$  
is complete for testing $R$, that is, 
$$
S\ \pass_\Leftrightarrow\  \TS_\Leftrightarrow       \Leftrightarrow S\models R
$$
holds for all $S\in {\cal D}$, that is, $\TS_\Leftrightarrow$ is a complete   suite for
testing composite requirement $R$.
\end{theorem}
\begin{proof}
Let $\ol x\in \Sigma_I^*$ be any nonempty input sequence and   $\ol y = \omega_S(\ii s, \ol x)$.
Let $\ell = |\ol x|=|\ol y|$.
We prove the following derivation for an arbitrary test case $\ol x$ of the complete
reduction test suite $\TS$ and its pass criterion $\pass_\Leftrightarrow$.
\begin{eqnarray}
 & & S\ \pass_\Leftrightarrow\  \ol x 
 \label{eq:z}
 \\
 & \Leftrightarrow & \forall j\in \{1,\dots,\ell\}:\ol y(j)\in \omega_{M_1'}(\ul q\after \ol x^{[1..j-1]}, \ol x(j))
 \label{eq:a}
\\
& \Leftrightarrow & \forall j \in \{1,\dots,\ell\}, \big(\ul q\after \ol x^{[1..j-1]},\ol x(j)\big)\in \{(q_i,x_i)~|~i=1,\dots,k\}: \nonumber
\\ & &  
\ol y(j)\in \omega_{M_1'}(\ul q\after \ol x^{[1..j-1]}, \ol x(j))
 \label{eq:b}
\\
& \Leftrightarrow & S\ \pass_\Leftrightarrow\  \text{pref}(\ol x)\cap \ol \Pi
\label{eq:c}
\end{eqnarray}
The equivalence \eqref{eq:z} $\Leftrightarrow$ \eqref{eq:a}   follows from re-writing the original definition of $\pass_\Leftrightarrow$ with explicit indexes for input and output events. 
Implication \eqref{eq:a} $\Rightarrow$ \eqref{eq:b} is trivial, since we restrict the 
$\ol y(j)$ under consideration to those that are outputs obtained when applying 
requirements-related inputs $x_i$ in state $q_i$. Conversely, \eqref{eq:a} $\Leftarrow$ \eqref{eq:b} follows from the fact that for 
$\big(\ul q\after \ol x^{[1..j-1]},\ol x(j)\big)\not\in \{(q_i,x_i)~|~i=1,\dots,k\}$, 
we have $\omega_{M_1'}(\ul q\after \ol x^{[1..j-1]}, \ol x(j)) = \Sigma_O$, so any output is acceptable. Finally, equivalence \eqref{eq:b} $\Leftrightarrow$ \eqref{eq:c} follows
from the fact that  \eqref{eq:b} is just the definition of  
$S\ \pass_\Leftrightarrow\  \text{pref}(\ol x)\cap \ol \Pi$ with explicit indexes.

With the  derivation above, we have shown that in order to prove $L(S) \subseteq L(M_1')$, it suffices
to check just the test cases   $\text{pref}(\ol x)\cap \ol \Pi$ for all $\ol x$ in the 
original complete test suite $\TS$. This proves 
$$
L(S) \subseteq L(M_1') \Leftrightarrow S\ \pass_\Leftrightarrow\ \TS_\Leftrightarrow.
$$
Now Theorem~\ref{th:reqreduct} can be applied to conclude that 
$S\models R\Leftrightarrow S\ \pass_\Leftrightarrow\ \TS_\Leftrightarrow$, and this completes the proof.
\qed
\end{proof}

\begin{example}\label{ex:reduction}
Applying the adaptive state counting algorithm from~\cite{hierons_testing_2004} to
reference model $M_1'$ from Example~\ref{ex:fsm_abs_nondet} and implementation model $S$ from Example~\ref{ex:runninga}, a test suite
$\TS$ with 39 test cases of maximal length 
 is obtained to test for reduction of $S$ against $M_1'$ with $n = 2$
for the number of states in the minimised observable FSM associated with $M_1'$, and 
$m = 3$ for the number of states in the minimised version of $S$ (which is identical to $S$).

Restricting this complete reduction test suite to the test cases that are also contained in
$\text{pref}(TS) \cap \ol \Pi$, where $\ol\Pi$ is specified according to Theorem~\ref{th:sre} 
for reference model $M$ and
requirement $R$ specified in Example~\ref{ex:runninga}, results in 16 test cases
\[
\begingroup\setlength{\arraycolsep}{10pt}
\begin{array}{llll}
a.a.a  & a.b.b.b.b  & b.a.b.b.b.a & b.b.b.a.b.a \\
a.a.b.a.b  & b.a.a.b.a.a  &  b.b.a.a &  b.b.b.b.a\\
a.b.a.b.a  & b.a.b.a.a  & b.b.a.b.a.b  &   b.b.b.b.b.b\\
a.b.b.a  & b.a.b.a.b.b  &  b.b.a.b.b &  a.a.b.b 
\end{array}
\endgroup
\]
It is easy to see that $S$ passes this test suite when applying pass criterion 
$\pass_\Leftrightarrow$ and reference model $M_1'$. Moreover, the decrease of test cases   in comparison to the full reduction
test suite is significant. However, observing that the exhaustive test suite for checking requirements satisfaction
 according to Theorem~\ref{th:se} needs only 4 test cases (see Example~\ref{ex:exhaustive}), 
 shows at least for this example that complete requirements testing 
needs far more test cases than exhaustive requirements testing.
This observation will be discussed in more detail in the next section.
\xbox
\end{example}

\section{Complexity Considerations}
\label{sec:complexity}

\subsection{Maximal Length of Test Cases}

\paragraph{Exhaustive requirements testing.} Consider first the maximal length $\text{tcl}_{max}^{exh}$ 
of test cases for exhaustive requirements testing according to  
Theorem~\ref{th:se}. From the test suite specification in \eqref{eq:tsone} and \eqref{eq:tstwo}, we conclude that
\begin{eqnarray}
\text{tcl}_{max}^{exh} & \le & \text{`maximal length of traces in $V$'} + {} \label{eq:tclmax}\\
& & m-n+1 + {} \nonumber \\
& & \text{`maximal length of distinguishing traces $\gamma$'} \nonumber
\end{eqnarray}
When using state covers $V$ with minimal-length input traces, the latter are bounded by $n-1$, where
$n$ is the number of states in the prime machine of the reference model. Also, 
minimal-length distinguishing traces are bounded by $n - 1$. This gives us an upper bound 
\begin{equation}
\text{tcl}_{max}^{exh} \le n + m  - 1
\end{equation}
for the maximal test case length. For Example~\ref{ex:exhaustive}, the state cover has traces of maximal length
1, and the minimal-length distinguishing traces have length 1.   Moreover, $m=n=3$. Applying Formula~\eqref{eq:tclmax}
results in $\text{tcl}_{max}^{exh} = 3$, and  this is confirmed by the test traces calculated in 
Example~\ref{ex:exhaustive} that are all of length 3.

\paragraph{Complete requirements testing.}
It is well known that testing for language equivalence requires   shorter test cases than
 testing for reduction. This is discussed, for example, in the lecture notes~\cite{PeleskaHuangLectureNotesMBT}:
 there, it is shown that for certain reference models, any complete reduction test suite needs test cases of maximal length 
 $$
 \text{tcl}_{max}^{cmp} = m\cdot n,
 $$
 where $m$ is the maximal number of SUT states, and $n$ the number of states  in the minimised observable 
 reference model~\cite[Section~4.5]{PeleskaHuangLectureNotesMBT}. This is reflected by 
 Example~\ref{ex:reduction}, where
 $m = 3$ and $n = 2$ (number of states in the prime machine of $M_1'$), 
 and the longest test case has $m\cdot n = 6$ inputs.

\subsection{Maximal Number of Test Cases}

When calculating   upper bounds for the number of test cases needed in an exhaustive or complete test suite,  test cases that are prefixes of others can be removed from the suite: if
an input traces $\ol x_1$ reveals an error in the implementation, then this error will also
be revealed by any longer input traces $\ol x_1.\ol x_2$ which has $\ol x_1$ as prefix. The bounds presented here take this observation into account.

\paragraph{Exhaustive requirements testing.} The maximal number  $\text{tc}_{max}^{exh}$ of test cases for exhaustive requirements testing coincides with the maximal number of test cases needed for language equivalence testing, since, when choosing requirement
$R_{eq}$ specified in Theorem~\ref{th:languageeq}, this characterises language equivalence. 

As a consequence, the estimate $\text{tc}_{max}^{exh}$ for exhaustive requirements testing
according to Theorem~\ref{th:se}
is that of the H-method. As pointed out by the inventors of the H-Method in~\cite{DBLP:conf/forte/DorofeevaEY05}, the upper bound depends on the implementation technique for the method in a critical way. Based on experiments made in~\cite{DBLP:conf/forte/DorofeevaEY05} and on our implementation in the fsmlib-cpp, the H-Method usually requires significantly fewer 
test cases than the well-known W-Method, for which the upper bound 
$n^2\cdot|\Sigma_I|^{m-n+1}$  is well-known~\cite{vasilevskii1973,chow:wmethod}.
As a consequence, it is safe to assume that test suites created by the H-Method fulfil
\begin{equation}\label{eq:tcmaxexh}
\text{tc}_{max}^{exh} \le n^2\cdot|\Sigma_I|^{m-n+1}
\end{equation}

\paragraph{Complete requirements testing.} For complete requirements testing according to Theorem~\ref{th:reqreduct} and Theorem~\ref{th:sre}, the maximal number  $\text{tc}_{max}^{cmp}$ of test cases required 
is the maximal number required for reduction testing, where the reduction test suite is generated from 
the nondeterministic abstraction $M_1'$ created from the original reference model $M$ and the requirement $R$ 
as described in Section~\ref{sec:m1prime}. 
We assume that the minimised equivalent of $M_1'$ has $n$ states.
Any complete reduction testing strategy can be used for this purpose, and from the resulting test suites $\TS$, all input traces outside $\text{pref}(\TS)\cap \ol\Pi$ can be removed, as 
shown in Theorem~\ref{th:sre}.

A very basic strategy derived from an investigation of product automata shows that  the set $\Sigma_I^{nm}$ of  all input sequences of length $m\cdot n$ is a complete reduction test suite (see, for example, the lecture notes~\cite[Section~4.5]{PeleskaHuangLectureNotesMBT}).

For most practical examples, (adaptive) state counting methods as published in~\cite{petrenko_testing_2011,DBLP:conf/hase/PetrenkoY14,hierons_testing_2004} need significantly fewer test cases that $|\Sigma_I^{nm}|$. In the general case, however, they may perform even worse than this bound derived from product automata. Fortunately, 
Theorem~\ref{th:sre} deals with deterministic implementations only, and the specific structure of the nondeterministic
reference models $M_1'$ introduced in Section~\ref{sec:m1prime} guarantees that every state is \emph{deterministically reachable} in the sense that we can calculate state covers as in the deterministic case, where every input
trace is guaranteed to lead to the specified target state -- only the outputs accompanying this input trace may vary nondeterministically. As a consequence, there exists a   deterministic state cover $V$ of $M_1'$ with $n$ elements whose traces  have bounded length
less or equal to $n-1$.
A standard argument from state counting methods (see references above) now implies that
test suites of the form
\[ 
V.\bigcup_{i=0}^{mn-n + 1} \Sigma_I^i
\]
are complete for checking language inclusion.
Observing that $|V| = n$ and prefixes of other test cases can be removed from a   test suite
without impairing its completeness properties, this results in the upper bound
\begin{equation}\label{eq:tcmaxcmp}
\text{tc}_{max}^{cmp} = n\cdot |\Sigma_I|^{mn - n + 1},
\end{equation}
for the number of test cases needed to check whether the SUT is a reduction of~$M_1'$.

\subsection{Discussion of Bounds for the Number of Test Cases}
 
Comparing the dominating values in formulas \eqref{eq:tcmaxexh} and \eqref{eq:tcmaxcmp},
we find that the exhaustive strategy is only exponential in the difference $m-n$, whereas the complete strategy is exponential in the product $m\cdot(n-1)$. This confirms 
the observation from Example~\ref{ex:reduction} that complete requirements testing 
needs considerably more test cases that the exhaustive strategy. 

From a practical perspective, it will be useful in most situations to learn additional errors about violations of language equivalence, even if they do not represent violations of requirements. The disadvantage of having to debug whether a failed test case points to a requirements violation or ``only'' to a general violation of language equivalence seems of lesser importance to us than the fact that the exhaustive strategy needs fewer test cases.

\section{Related Work}\label{sec:related}

The use of formal specifications, in particular, reference models with formal behavioural semantics, 
has a long tradition in testing~\cite{DBLP:journals/csur/HieronsBBCDDGHKKLSVWZ09,Petrenko:2012:MTS:2347096.2347101,DBLP:journals/jss/AnandBCCCGHHMOE13}. Among the numerous formal approaches, complete test strategies have received special attention, because they guarantee full fault coverage with respect to a reference model and a conformance relation under certain well-defined hypotheses concerning the potential faults of the system under test (SUT). Complete strategies are of particular interest in the domain of safety-critical systems, where a justification of the test case selection is required in order to obtain certification credit. 
These strategies have been  comprehensively investigated in the context of  conformance testing.
Typical conformance relations were I/O-language equivalence or language containment~\cite{vasilevskii1973,chow:wmethod,DBLP:conf/forte/DorofeevaEY05,DBLP:conf/forte/DorofeevaEY05,simao_reducing_2012,hierons_testing_2004,petrenko_testing_2011,peleska_sttt_2014}, 
refinement relations for process algebras 
and related formalisms~\cite{DBLP:conf/icfem/CavalcantiG07,DBLP:journals/acta/CavalcantiG11},   
and the well-known ioco-relation~\cite{tretmans1996}. In the context of hybrid systems, new conformance 
relations have been proposed, for example, in~\cite{DBLP:journals/scp/AraujoCMMS18}.

Theorem~\ref{th:sre} states that {\it any} complete reduction test suite can be modified
to yield a (usually smaller) complete suite for requirements testing. Typical complete
reduction testing strategies are  based on a \emph{state counting} method; these are essential for
testing reduction relations  between nondeterministic reference models and (deterministic or nondeterministic) SUTs. State counting has been explained in~\cite{hierons_testing_2004,petrenko_testing_2011}. The fact that reduction (in our case, the reduction of a model abstraction $M_1'$ presented in Section~\ref{sec:m1prime}) preserves
requirements satisfaction is a very general insight which holds for different modelling formalisms
and requirements specification methods. 
We refer here the well-known fact that LTL specifications can be checked by a maximally nondeterministic 
Buchi automaton and all refinements thereof fulfil the   formula as well~\cite{DBLP:books/daglib/0020348}. The Unified Theories of Programming~\cite{hoare1998} investigate both requirements
satisfaction and refinement 
on a more general logical level and provide the insight that refinement is 
strongly related to logical implication. Therefore,   requirements satisfaction is preserved by refinement as a 
simple logical consequence.

While conformance testing is preferred in the field of protocol verification~\cite{protocoltestsystems95,DBLP:conf/pts/TretmansKB91}, other application areas follow the property-driven 
approach, where it has to be established that the SUT implements a collection of 
requirements in  a correct way~\cite{machado_towards_2007,DBLP:conf/fates/FernandezMP03,DBLP:conf/soqua/LiQ04,reqbasedtestingskokovic,DBLP:conf/isola/0001BH18}. Again, models reflecting the requirements under consideration may be used. Alternatively, implicit specifications in linear
 temporal 
logic can be constructed to identify the I/O-traces fulfilling a requirement; the underlying theory has been elaborated, for example, in \cite[Section~4.2]{DBLP:books/daglib/0020348} and \cite{Safra:1988:COA:1398513.1398627}. Since logic specifications referring to the interfaces of the SUT alone 
can become very complex, it is often advocated to combine models with temporal logic specifications, so that the latter may also refer to internal states of the model, thereby simplifying the formulae~\cite{DBLP:conf/ets/Peleska18,DBLP:conf/isola/0001BH18}.

When applied to complex real-world systems, 
complete FSM-based testing methods usually require very large test suites.
This problem    has been mitigated in recent years by abstraction techniques based on equivalence classes and symbolic state machines, see, for example, \cite{Huang2017,DBLP:conf/models/Petrenko18}: the original FSM-based test generation algorithms can be applied to abstractions of more complex models, such as extended finite state machines or UML state machines. Examples presented in~\cite{DBLP:conf/rssrail/PeleskaHH16}  show that these abstraction techniques allow for complete testing of quite complex control systems with feasible effort.


\section{Conclusion}\label{sec:conc}

In this paper, a notion for specifying elementary and composite requirements in deterministic finite state machine models has been defined. Two   black-box 
testing strategies proving or disproving that an implementation satisfies these requirements have been presented. The first is exhaustive in the sense that every implementation
violating the requirement will fail at least one test case generated according to the strategy. Failing a test case implies that the SUT is not language equivalent to the reference model, but may not necessarily mean that the implementation violates the requirement. The second is complete in the sense that it is exhaustive and guarantees that 
failing a test case alway implies that the system under test violates the requirement.
The exhaustiveness and completeness properties, respectively,  of the test suites generated according to these strategies hold under the assumption that the implementation has no more 
than $m \ge n$ states, where $n$ is the known number of states in the reference model.

The implementation of the first strategy is based on the H-method.
Using a real-world application, it is demonstrated that the first strategy frequently requires significantly less test cases than a complete method establishing language equivalence 
by means of the original H-method. Therefore, the new method is well-suited for testing 
a selection of critical requirements with guaranteed fault coverage, while less critical requirements can be tested in the conventional way, using heuristics for test case generation.

The implementation of the second strategy is based on a state counting method which 
has been originally used to test for language inclusion. While the second strategy 
leads to smaller test suites in comparison to complete methods showing language inclusion, it usually results in more test cases than needed for establishing language equivalence. Consequently, it is of more theoretical interest, at least if a model can be constructed 
that is equivalent to the desired behaviour of the implementation.


%


\begin{acknowledgements}
The authors would like to thank Robert Sachtleben
for pointing out important details about the complexity of   state counting methods.
\end{acknowledgements}




\appendix
\normalsize
\section{Tool Support and Resources}
\label{sec:fsmlib}

The test suites presented in the examples above have been generated using the 
open source library
fsmlib-cpp. The library is downloaded, compiled, and some executables are created according to the instructions given in~\cite[Appendix B]{PeleskaHuangLectureNotesMBT}. After that, an executable {\tt fsm-generator} is available which allows for test suite generation according to various strategies without having to write own main programs accessing the library classes and their methods.

Reference model $M$ and SUT model $S$ from Example~\ref{ex:runninga} can be found in the
fsmlib-cpp installation, directory
\footnotesize
\begin{verbatim}
resources/complete-requirements-based-testing/Example-2-3-4
\end{verbatim}
\normalsize
as files {\tt M.csv}  and {\tt S.csv}, respectively (deterministic FSMs can be specified 
in CSV format as explained in~\cite[Appendix B]{PeleskaHuangLectureNotesMBT}).   

To re-generate the H-method test suite checking language equivalence as described in Example~\ref{ex:exhaustive}, change into
the Example-2-3-4-directory and type command
\footnotesize
\begin{verbatim}
<path-to-executable>/fsm-generator -h -a 0 M.csv
\end{verbatim} 
\normalsize
Option {\tt -a} specifies the maximal number of additional states allowed in the SUT, so this generation creates an H-test suite from $M$ which is complete for language equivalence testing under the assumption that $S$ does not have more states than $M$.

For generating the exhaustive requirements test suite from $M$ and $R$ according to
Example~\ref{ex:exhaustive}, the abstraction $M_1$ has to be manually created from $M$ and $R$
as specified in Example~\ref{ex:runninga}. This DFSM is also stored in the Example-2-3-4 directory as file {\tt M1.csv}. The exhaustive requirements test suite  is now generated by command
\footnotesize
\begin{verbatim}
<path-to-executable>/fsm-generator -s -h -a 0 M.csv M1.csv
\end{verbatim}
\normalsize
The {\tt -s} parameter specifies requirements-based testing, and in such a case, a second
FSM specification file (here: {\tt M1.csv}) is expected as parameter. This call to the generator creates exactly the test suite with the 4 test cases shown in  Example~\ref{ex:exhaustive} for requirements-based testing according to Theorem~\ref{th:se}.

For generating the test suites related to the Fasten-Seat-Belt and Return-to-Seat-Sign controller described in Section~\ref{sec:rwexample}, change into directory
\footnotesize
\begin{verbatim}
resources/complete-requirements-based-testing/Section-6
\end{verbatim}
\normalsize
The reference model described in Section~\ref{sec:rwexample} is stored in this directory
as {\tt FSBRTSX.csv} (note that the csv-file uses other state names than the $s_i$ shown in Table~\ref{tab:fsb}). The DFSM abstraction created from requirement $\mathbf{R}_1$ is
contained in file {\tt FSBRTSX-ABS-R1.csv}. To re-create the test suites for different values of $m-n$ as specified in Table~\ref{tab:hversusimplication}, column $\mathbf{R}_1$, use commands
\footnotesize
\begin{verbatim}
<path-to-executable>/fsm-generator -s -h -a 0 FSBRTSX.csv FSBRTSX-ABS-R1.csv
<path-to-executable>/fsm-generator -s -h -a 1 FSBRTSX.csv FSBRTSX-ABS-R1.csv
<path-to-executable>/fsm-generator -s -h -a 2 FSBRTSX.csv FSBRTSX-ABS-R1.csv
\end{verbatim}
\normalsize

For requirement $\mathbf{R}_2$ specified in Section~\ref{sec:rwexample}, use abstraction file
file {\tt FSBRTSX-ABS-R2.csv} and commands that are equivalent to the ones shown above. To create the complete 
test suites for language equivalence with the H-Method, use commands
\footnotesize
\begin{verbatim}
<path-to-executable>/fsm-generator -h -a 0 FSBRTSX.csv 
<path-to-executable>/fsm-generator -h -a 1 FSBRTSX.csv 
<path-to-executable>/fsm-generator -h -a 2 FSBRTSX.csv 
\end{verbatim}
\normalsize

\end{document}